\documentclass[letterpaper,11pt]{article}
\usepackage[utf8]{inputenc}

\usepackage{amsmath, amsthm, amssymb, thm-restate}
\usepackage{algorithmicx}
\usepackage[table,xcdraw]{xcolor}
\usepackage[vlined,linesnumbered]{algorithm2e}
\usepackage{setspace}
\usepackage{mathtools}
\usepackage[numbers]{natbib}
\usepackage{comment} 
\usepackage{tcolorbox} 
\usepackage{xfrac}
\usepackage{hyperref}
\usepackage{multirow}
\usepackage{caption}
\usepackage{bm}
\usepackage{newfloat}
\usepackage{enumitem}
\usepackage{dblfloatfix} 
\usepackage{wrapfig}
\usepackage{dsfont}
\usepackage{mdframed}
\usepackage{multirow}

\usepackage{fancyhdr}
\usepackage[noabbrev]{cleveref}

\usepackage[margin=1in]{geometry}

\allowdisplaybreaks

\definecolor{mygreen}{RGB}{10,110,230}
\definecolor{myred}{RGB}{10,110,230}

\hypersetup{
     colorlinks=true,
     citecolor= mygreen,
     linkcolor= myred
}

\renewcommand{\epsilon}{\varepsilon}

\DeclareMathOperator{\E}{\ensuremath{\normalfont \textbf{E}}}

\newcommand{\hiddencomment}[1]{}

\newcommand{\mc}[1]{\ensuremath{\mathcal{#1}}}

\newcounter{protocolcounter}
 
\crefname{protocolcounter}{Algorithm}{Algorithms}
\newcommand{\unmapr}[0]{\ensuremath{\delta}}
\newcommand{\opt}[0]{\text{OPT}}
\newcommand{\MM}[0]{{\textsf{MWM}}}
\newcommand{\weight}[0]{{\sf{W}}}
\newcommand{\adj}{\ensuremath{N}\xspace}

\newcommand{\SX}{\ensuremath{X}\xspace}
\newcommand{\funcZ}[0]{f}
\newcommand{\VB}[0]{\ensuremath{\mathcal{VB}}}
\newcommand{\apx}[0]{\ensuremath{\mathcal{\alpha}}}

\SetKwRepeat{Do}{do}{while}

\renewcommand{\b}[1]{\ensuremath{\bm{\mathrm{#1}}}}
\DeclareMathOperator{\poly}{poly}
\DeclareMathOperator{\polylog}{polylog}

\renewcommand{\O}[1]{\ensuremath{O\left(#1\right)}}

\DeclareMathOperator{\var}{Var}

\renewcommand{\epsilon}[0]{\ensuremath{\varepsilon}}

\let\originalleft\left
\let\originalright\right
\renewcommand{\left}{\mathopen{}\mathclose\bgroup\originalleft}
\renewcommand{\right}{\aftergroup\egroup\originalright}

\newtheorem{theorem}{Theorem}[section]

\newtheorem{lemma}{Lemma}[section]

\newtheorem{proposition}[lemma]{Proposition}

\newtheorem{definition}[lemma]{Definition}
\newtheorem{claim}[lemma]{Claim}

\makeatletter
\def\thm@space@setup{%
  \thm@preskip= 0.2cm
  \thm@postskip=\thm@preskip 
}
\makeatother

\definecolor{mygreen}{RGB}{20,155,20}
\definecolor{myred}{RGB}{195,20,20}
\definecolor{linkcolor}{RGB}{0,0,230}
\definecolor{mylightgray}{RGB}{230,230,230}
\definecolor{verylightgray}{RGB}{240,240,240}
\definecolor{commentcolor}{RGB}{120,120,120}

\SetCommentSty{mycommfont}

\hypersetup{
     colorlinks=true,
     citecolor= mygreen,
     linkcolor= myred,
     urlcolor= mygreen
}

\renewcommand{\mc}[1]{\ensuremath{\mathcal{#1}}}

\newcounter{myalgctr}

\newenvironment{tbox}{
\par\addvspace{0.2cm}
\begin{tcolorbox}[width=\textwidth,
                  boxsep=2pt,
                  left=1pt,
                  right=1pt,
                  top=4pt,
                  boxrule=1pt,
                  arc=0pt,
                  colback=white,
                  colframe=black
                  ]
}{
\end{tcolorbox}
}

\newenvironment{tboxh}{
\par\addvspace{0.2cm}
\begin{tcolorbox}[width=\textwidth,
                  boxsep=2pt,
                  left=1pt,
                  right=1pt,
                  top=4pt,
                  boxrule=1pt,
                  arc=0pt,
                  colback=white,
                  colframe=black,
                  float=t
                  ]
}{
\end{tcolorbox}
}

\newcommand{\tboxhrule}[0]{\vspace{0.1cm} \hrule \vspace{0.2cm}}

\newenvironment{titledtbox}[1]{\begin{tbox}#1 \tboxhrule}{\end{tbox}}
\newenvironment{titledtboxh}[1]{\begin{tboxh}#1 \tboxhrule}{\end{tboxh}}

\newenvironment{tboxalg2e}[1]{
\refstepcounter{myalgctr}
	\begin{titledtbox}{\textbf{Algorithm \themyalgctr.} #1}
	\vspace{-0.2cm}
}
{
	\vspace{-0.3cm}
	\end{titledtbox}
}

\newcommand{\restateclaim}[2]{\noindent \textbf{Claim~#1} (restated) \textbf{.} {\em #2}}

\newcommand{\restatelem}[2]{\noindent \textbf{Lemma~#1} (restated) \textbf{.} {\em #2}}

\makeatletter
\renewcommand{\paragraph}{%
  \@startsection{paragraph}{4}%
  {\z@}{10pt}{-1em}%
  {\normalfont\normalsize\bfseries}%
}
\makeatother

\makeatletter
\patchcmd{\@algocf@start}
  {-1.5em}
  {0pt}
  {}{}
\makeatother

\title{Query Efficient Weighted Stochastic Matching}

\author{
Mahsa Derakhshan \\ {\em Northeastern University} \and Mohammad Saneian  \\ {\em Northeastern University}
}

\date{}

\begin{document}
\maketitle

\thispagestyle{empty}
In this paper, we study the {\em weighted stochastic matching} problem. Let $G=(V, E)$ be a given edge-weighted graph and let its realization $\mathcal{G}$ be a random subgraph of $G$ that includes each edge $e\in E$ independently with a known probability $p_e$. The goal in this problem is to pick a sparse subgraph $Q$ of $G$ without prior knowledge of $G$'s realization, such that the maximum weight matching among the realized edges of $Q$ (i.e. the subgraph $Q\cap \mc{G}$) in expectation approximates the maximum weight matching of the entire realization $\mc{G}$.

It is established by previous work that attaining any constant approximation ratio for this problem requires selecting a subgraph of max-degree $\Omega(1/p)$ where $p=\min_{e\in E} p_e$. On the positive side, there  exists a $(1-\epsilon)$-approximation algorithm by Behnezhad and Derakhshan [FOCS'20],  albeit at the cost of max-degree having exponential dependence on $1/p$.
Within the $\poly(1/p)$  regime, however, the best-known algorithm  achieves a $0.536$ approximation ratio due to Dughmi, Kalayci, and Patel [ICALP'23] improving over the $0.501$ approximation algorithm by Behnezhad, Farhadi, Hajiaghayi, and Reyhani [SODA'19].

In this work, we present a 0.68 approximation algorithm with $O(1/p)$ queries per vertex, which is asymptotically tight. This is even an improvement over the best-known approximation ratio of $2/3$ for unweighted graphs within the $\poly(1/p)$ regime due to Assadi and Bernstein [SOSA'19]. The $2/3$ approximation  ratio is proven tight in the presence of a few correlated edges in $\mathcal{G}$, indicating that surpassing the $2/3$ barrier  should rely on the independent realization of edges. Our analysis involves reducing the problem to designing a randomized matching algorithm on a given stochastic graph with some \emph{variance-bounding} properties. 
Both the reduction and the existence of such algorithms heavily rely on independent edge realizations, allowing us to break the 2/3 barrier.


\clearpage
\setcounter{page}{1}
\section{Introduction}
In the {\em stochastic weighted  matching} problem, we are given an $n$-vertex weighted graph $G=(V, E)$ along with a parameter $p_e \in (0, 1]$ for any edge $e\in  E$. A random subgraph $\mc{G}$ of $G$ is then generated by independently including (or realizing) each edge $e \in E$ with  probability  $p_e$. Here, we refer to $G$ as the base graph and $\mc{G}$ as the realized subgraph. 
The objective of this problem is to select a subgraph $Q$ of the base graph without the knowledge of its realization  such that: (1) $Q$ has a small max-degree, namely a constant with respect to $n$, and (2) The realized edges of $Q$ (i.e., the graph $Q \cap \mc{G}$) contain a large weight approximate matching.
We define the approximation ratio as the expected weight of the maximum matching among the realized edges of $Q$ over the expected weight of the maximum weighted matching of $\mc{G}$. 

One immediate application of the stochastic weighted matching problem is its use as a {\em matching sparsifier}, which approximates the maximum weighted matching even when random edge failures occur \cite{sosa19}. Additionally, it finds various applications in matching markets, including kidney exchange~\cite{blumetal}, online labor markets \cite{soda19, BR18}, and dating platforms. In these applications, 
we are  provided with the base graph $G$, but we are tasked with finding a matching in the realized subgraph $\mc{G}$. To achieve this, an algorithm can {\em query} each edge of $G$ to determine whether it is realized. However, these queries often involve time-consuming or costly operations, such as conducting candidate interviews or medical exams. Hence, it is crucial to minimize the number of queries. This can be accomplished by non-adaptively querying a subgraph $Q$ with a small degree while still expecting to find a matching with a large approximation ratio among its realized edges.

It is established by previous work~\cite{AKL16} that attaining any constant approximation ratio for this problem requires selecting a subgraph of max-degree $\Omega(1/p)$ where $p=\min_{e\in E} p_e$. On the positive side, there  exists a $(1-\epsilon)$-approximation algorithm by Behnezhad and Derakhshan~\cite{behnezhad2020stochastic},  albeit at the cost of max-degree having exponential dependence on $1/p$.
Within the $\poly(1/p)$  regime, however, the best-known algorithm  achieves a $0.536$ approximation ratio due to Dughmi, Kalayci, and Patel~\cite{DughmiKP23}.  This raises the question: How far can we push the approximation ratio while querying a graph with max-degree $\poly(1/p)$? A natural barrier that emerges here is  $2/3$ which cannot be exceeded in the presence of a few correlated edges in $\mc{G}$~\cite{AKL16, stoc20}.  Existing algorithms that attain better than a 2/3 approximation ratio, even for unweighted graphs, either require a max-degree exponential in $1/p$~\cite{stoc20} or heavily rely on the bipartiteness of the graph~\cite{behnezhad2022stochastic}. 

In this work, we prove the existence of a subgraph $Q$ with max-degree $O(1/p)$ whose realization contains a $0.681$ approximate matching. This not only improves the previously known approximation ratio of $0.536$ significantly but also break the $2/3$ barrier in the $\poly(1/p)$ regime.

\newcommand{\thmmain}[0]{
There exists an algorithm that picks a $\O{1/p}$ degree subgraph $Q$ of $G$ such that the expected weight of the max-weight realized matching in $Q$ is at least $0.681$ times the expected weight of the max-weight realized matching in $G$.
}
\begin{mdframed} [backgroundcolor=gray!20]
\begin{theorem}
\label{thm:Themain}
\thmmain
\end{theorem}
\end{mdframed}

The significance of our result is threefold. First, we  improve the best-known approximation ratio for both weighted and unweighted graphs  in the $\poly(1/p)$ regime (respectively from $0.536$~\cite{DughmiKP23} and $2/3$~\cite{sosa19}). Second, it selects a subgraph with max-degree $\O{1/p}$  --- a bound that is asymptotically tight for achieving any constant approximation ratio. Third, we break the well-established $2/3$ barrier for the approximation ratio within the $\poly(1/p)$ regime. Additionally, we demonstrate that the problem can be reduced to designing approximate matching algorithms with specific properties, which we term {\em variance-bounding} matching algorithms. This reduction implies that further exploration of the stochastic matching problem may be focused on the development of such algorithms.

\paragraph{Further Related Work.} After the pioneering work of Blum et al.~\cite{blumetal}, the stochastic matching problem has received considerable attention~\cite{blumetal,AKL16,AKL17,YM18,BR18,soda19,sosa19, d2019stochastic, stoc20, behnezhad2020stochastic}. To provide a comprehensive understanding of the research landscape, Table~\ref{table1} presents a chronological survey of results for the problem considered in this paper. 

\begin{table}[htbp]
\centering
\renewcommand{\arraystretch}{1.1}

\begin{tabular}{|l|p{8.75cm}|p{2.53cm}|p{3.14cm}|}
\hline
\rowcolor[HTML]{C0C0C0} 
& \textbf{Reference} & \textbf{Approx.} & \textbf{Max-degree of $Q$} \\
\cline{2-4}
\hline
& Blum, Dickerson, Haghtalab, Procaccia, Sandholm, Sharma~\cite{blumetal} & $0.5-\epsilon$ & $\poly(1/p)$ \\
\cline{2-4}
& Assadi, Khanna \& Li~\cite{AKL17} & $0.5001$ & $\O{\log(1/p)/p}$ \\
\cline{2-4}
& Behnezhad, Farhadi, Hajiaghayi, Reyhani~\cite{soda19} & $0.6568$ & $\O{\log(1/p)/p}$ \\
\cline{2-4}
& Assadi \& Bernstein~\cite{sosa19} & $2/3-\epsilon$ & $\O{\log(1/p)/p}$ \\
\cline{2-4}
& Behnezhad, Derakhshan \& Hajiaghayi~\cite{stoc20} & $1-\epsilon$ & $(1/p)^{\polylog (1/p)}$ \\
\cline{2-4}
\multirow{-6}{*}{\rotatebox[origin=c]{90}{Unweighted}} & Behnezhad, Blum \& Derakhshan~\cite{behnezhad2022stochastic} & $0.73$ (bipartite) & $\O{\log(1/p)/p}$ \\
\hline
\cline{2-4}
& Yamaguchi \& Maehara~\cite{yamaguchi2018stochastic} & $0.5 -  \epsilon$ & $O(W \log n/p)$ \\
\cline{2-4}
& Behnezhad \& Reyhani~\cite{BR18} & $0.5-\epsilon$ & $\poly(1/p)$ \\
\cline{2-4}
& Behnezhad, Farhadi, Hajiaghayi,  Reyhani~\cite{soda19} & $0.501$ & $\poly(1/p)$ \\
\cline{2-4}
& Behnezhad \& Derakhshan~\cite{behnezhad2020stochastic} & $1-\epsilon$ & $ \exp{(1/p)}$ \\
\cline{2-4}
& Dughmi, Kalayci \& Patel~\cite{DughmiKP23} & $0.536$ & $O(1/p)$ \\
\cline{2-4}
\multirow{-6}{*}{\rotatebox[origin=c]{90}{Weighted}} & \textbf{This work} & \bm{$0.681$} & $\bm{O(1/p)}$ \\
\hline
\end{tabular}
\caption{Survey of known results in chronological order. For simplicity, we have hidden the actual dependence on $\epsilon$ inside the $O$-notation in some cases.}
\label{table1}
\end{table}
\vspace{-1 cm}
\section{Technical Overview}\label{sec:techniques}
The algorithm we use to construct $Q$ is a quite simple one which was  introduced by~\cite{soda19} and subsequently studied by~\cite{stoc20,behnezhad2022stochastic}. Given  a parameter $t$, the algorithm starts by drawing  $t$ realizations of $G$ drawn from the same distribution as $\mathcal{G}$. Let us represent these random subgraphs by $\mc{G}_1, \dots, \mc{G}_t$. We then let $Q$ be the union of max-weight matchings of these graphs. That is $$Q:=\cup_{i\in t} \MM(\mc{G}_i),$$ where $\MM(\cdot)$ is a deterministic algorithm returning the max-weight matching of a given graph. Since $Q$ is a union  of $t$ matchings, it clearly has max-degree $t$. The challenge, however, is proving that for a $t$ as small as $\O{1/p}$, the realization of $Q$ contains a large weight matching. We provide a constructive proof for this. That is,  we design  an algorithm for finding a matching with a large approximation ratio in $\mc{Q}$, the actual realization  of $Q$. Below, we first briefly review the ideas used by the previous work and then discuss the ingredients we add to achieve our desired result.

\paragraph{Crucial/Non-crucial Edge Decomposition}

The framework utilized to analyze the aforementioned algorithm involves partitioning the edges into two categories: \textit{crucial} and \textit{non-crucial}. Separate arguments are then presented to demonstrate how these edges can be integrated to construct a large weight enough matching in $\mathcal{Q}$. Let $x_e$ denote the probability of edge $e$ being part of the optimal solution, i.e., 
\[ x_e = \Pr[e \in \MM(\mathcal{G})]. \]
We define the set of crucial edges, denoted by $C$, and the set of non-crucial edges, denoted by $N$, as follows:
\[ C = \{ e \in E : x_e \geq \tau \} \;\;\;\;\;\text{ and } \;\;\;\;\; N = \{ e \in E : x_e < \tau \},\]
where $\tau$ is a fixed threshold in the  order of $p$. Note that by choosing a sufficiently large value for $t=\O{1/p}$, we can ensure that $Q$ contains nearly all of the crucial edges. To establish the existence of a large weight matching in $\mc{G}$, the  first step  is to construct a matching $M_c$ exclusively on the crucial edges which is an $\alpha$-approximation  with respect to the contribution of  the crucial edges to the  optimal solution. ($M_c$ should  satisfy some other useful properties which we will discuss later.)  The next step is constructing a  fractional matching $\bm{f}$ on the subgraph of non-crucial edges whose endpoints are unmatched in $M_c$. This fractional matching should satisfy the two following properties: first, for any edge $e\in N$, it holds that $\E[f_e]\simeq x_e$; second, the values of $f_e$ should be small enough to ensure that $\bm{f}$ has almost no integrality gap.
By combining these steps, the framework constructs a matching with weight almost $\alpha\times \weight(\MM(\mathcal{G})\cap C)+ \weight(\MM(\mathcal{G})\cap N)$. Here $\weight(.)$ is a function returning the weight of a given matching.

All papers utilizing this analysis framework require the algorithm used for constructing $M_c$ to match the endpoints of any non-crucial edge $e$ independently. Otherwise, the edge is discarded. This requirement is the main reason why Behnezhad and Derakhshan  \cite{behnezhad2020stochastic} need to take an exponential number of edges per vertex. To achieve this property, they employ a distributed LOCAL algorithm for constructing $M_c$, which can lead to a vertex being dependent on the vertices within its $\Omega(\log(\Delta))$ radius ball, where $\Delta$ denotes the crucial degree of a vertex. Since potentially $(1/p)^{\log(1/p)}$ non-crucial edges may be discarded for each vertex, these edges need to have small $x_e$ values. Consequently, a small threshold $\tau$ and a large $t$ must be chosen.
Due to known lower bounds for matching in  the LOCAL model~\cite{kuhn2016local}, one  cannot hope to prove desirable approximation ratios for a $Q$ of max-degree $\poly(1/p)$ following this approach.

In this work, we demonstrate that it is possible to relax the requirement regarding the independent matching of endpoints of any non-crucial edge in $M_c$. Instead, we replace it with an upper bound on the variance of a parameter related to the neighborhood of each vertex. Specifically, it should be possible to pick a subset $A$ of the vertices unmatched by $M_c$ such that:

\begin{enumerate}
    \item Any non-crucial edge $e=(u,v)$  satisfies $\Pr[\{u,v\}\subset A] \geq \delta$ for a fixed constant $\delta>0$.
    
    \item Let us define random variable $Z_v$ related to the neighborhood of any vertex $v$ as \begin{equation}\label{eq:ejjrke}
         Z_v = \sum_{e=(u,v)\in N}\frac{x_e}{\Pr[\{u,v\}\subset A]}\mathds{1}_{u\in A}.
    \end{equation} 
    Note that the randomization here is due to $A$ and $M_c$ being random variables themselves. We require the variance of this random variable to be upper-bounded as follows: $\var(Z_v)\leq \frac{10\tau}{\delta^2}.$ 
\end{enumerate} 
We define an algorithm for finding $M_c$ and $A$ to be a {\em variance-bounding} matching algorithm (see definition~\ref{def:variance-bounding}) if it satisfies the above-mentioned property (and a few others). We then provide a reduction demonstrating  that if $M_c$ is an $\alpha$-approximation with respect to the contribution of crucial edges to the optimal solution, then it is possible to find an almost $\frac{1}{2-\alpha}$-approximate solution on the realized edges of $Q$. Our proof strongly relies on independent edge realizations hence enabling us to break the $2/3$ barrier. 

The second step of our analysis involves showing the existence of an $\frac{8}{15}$-approximate variance-bounding matching algorithm. (See Algorithm~\ref{alg:Fu-modified}.) We utilize a randomized algorithm designed by Fu et al.~\cite{Fu} for a variant of online stochastic matching and prove that, with slight modifications, it satisfies our desired property. Our analysis begins by identifying a set of independent random variables that determine the algorithm's outcome. We then utilize a method called the Efron-Stein inequality to establish the desired  upper bound on $\var(Z_v)$. Since this inequality is not very  commonly used in theoretical computer science, we hope for our work to serve as an example of using this powerful tool in the analysis of randomized algorithms.

\section{Preliminaries}
\subsection{Notation}
In the stochastic weighted matching  problem, the input is an $n$-vertex graph $G=(V,E)$, a vector of weights $\bm{w}=\{w_e: e\in E\}$ and a probability  vector $\bm{p}=\{p_e: p_e \in E\}$.  Subgraph $\mc{G}$ is a  random  subgraph of $G$ which contains each edge independently with probability $p_e$. The goal in this problem is to pick a subgraph $Q$ of $G$ without the knowledge of its realization  such that: (1) $Q$ has a small max-degree, namely a constant with respect to $n$, and (2) The realized edges of $Q$ (i.e., the graph $Q \cap \mc{G}$) contain a large weight approximate matching.
We define the approximation ratio as 
 $$\frac{\E\left[\weight\left (\MM(\mc{Q})\right)\right]}{\E[\weight(\MM(\mc{G}))]},$$
where $\mc{Q} = \mc{G}\cap Q$ is the realization of $Q$ and $\MM(.)$ is a deterministic algorithm returning a maximum weighted matching of a given graph, and $\weight(M) = \sum_{e\in M} w_e$ is a function returning the weight of a given matching $M$. We will use $\opt=\MM(\mc{G})$ to refer to  the maximum matching of the actual realization.  We may sometimes abuse notation and use $\opt$ to refer to its expected weight when it is clear from the context. Note that while $\opt$ is a random variable $\E[\weight(\opt)]$ is just a number. For any edge $e\in E$, we define $$x_e=\Pr[e\in \opt],$$ where the probability is taken over the randomization in $\mc{G}$. Similarly, for any vertex $v\in V$ we let $x_v=\Pr[v\in \opt]$ be the probability that $v$ is matched in $\opt$. By the definition stated below, $\bm{x}$ is a fractional matching as each vertex joins $\opt$ w.p. at most one.

\begin{definition}[Fractional matching]\label{def:frac-matching}
A \emph{fractional matching} $\bm{x}$ of a graph $G=(V, E)$ is an assignment $\{x_e\}_{e \in E}$ to the edges, where $x_e \in [0, 1]$ and for each vertex $v \in V$, $x_v := \sum_{e \ni v} x_e \leq 1$. We use $|\bm{x}| := \sum_e x_e$ to denote the size of a fractional matching, and for any subset $E' \subseteq E$, use $\b{x}(E')$ to denote $\sum_{e \in E'} x_e$. 
\end{definition}
Throughout the paper, we use the notation  $O_\epsilon(f(n))$ which means we have assumed $\epsilon$ to be a constant to calculate the complexity of $f(n)$. The max-degree of subgraph $Q$ we find in this paper depends on the smallest $p_e$ amongst all edges, which we refer to as $p$. In other words 
$$p=\min_{e\in E} p_e.$$

In the following table, we list a set of variables  and their values, which we will use throughout the paper. Values are defined as functions of $\epsilon\in (0,1)$, which is a sufficiently small constant, and $\delta\in (0,1)$, which we will introduce in Definition~\ref{def:variance-bounding}. 
\begin{table}[htbp] 
\centering

\setlength{\arrayrulewidth}{1pt}
\setlength{\tabcolsep}{15pt}
\renewcommand{\arraystretch}{1.5}
\begin{tabular}{|>{\columncolor{gray!50}}c|>{\columncolor{gray!50}}c|>{\centering\arraybackslash}m{1.8cm}|>{\centering\arraybackslash}m{1.8cm}|>{\centering\arraybackslash}m{1.8cm}|>{\centering\arraybackslash}m{1.8cm}|}
\hline
\rowcolor{white}
\textbf{Variable} & \textbf{$\tau$}                     & \textbf{$\eta$} & \textbf{$\beta$}         & \textbf{$\gamma$}    & \textbf{$c$}             \\ \hline
\rowcolor{white}
\textbf{Value}    & $20p\epsilon^5 \delta^2$ & $\epsilon/10$          & $\epsilon^2/100$ & $\frac{1-\epsilon^2}{1+3 \eta}$          & $10/\epsilon$ \\ \hline
\end{tabular}\caption{Value of the parameters used throughout the paper}
\label{table:values} 
\end{table}

\subsection{Concentration Inequalities and Probabilistic Tools}
In this section, we state the concentration inequalities and some of the probabilistic tools that will be used throughout the paper. 

\begin{proposition}[The Efron–Stein Inequality~\cite{steele1986efron}]\label{prop:efron-stein} 
    Suppose  $X_1,..., X_n, X_1',..., X_n'$ are independent random variables with $X_i'$ and $X_i$ having the same distribution for all $i$. 
    Let $X = (X_1, ..., X_n)$ and $X^{(i)} = (X_1,..., X_{i-1}, X_i', ..., X_{i+1}, ...,  X_n)$. Then: 
    $$
    \var(f(X)) \leq \frac{1}{2} \sum_{i=1}^n \E \left[\left(f(X) - f(X^{(i)}) \right)^2\right]
    $$
\end{proposition}

\begin{proposition}[Chebyshev's Inequality]\label{prop:Chebyshev}
Let $X$ be a random variable with finite non-zero standard deviation $s$, (and thus finite expected value $\mu$.) Then for any real number $c > 0$, we have
$$
\Pr[|X - \mu| \geq c s] \leq \frac{1}{c^2}.
$$ 
    
\end{proposition}

\begin{proposition}[Law of Total Variance]\label{prop:lawof}
Let $X$ be a random variable and $Y$ be a random variable with respect to the same sample space. Then, the variance of $X$ can be expressed as
\begin{equation*}
\var(X) = \E[\var(X\mid Y)] + \var(\E[X\mid Y]).
\end{equation*}
\end{proposition}

\begin{definition}[Negative Association]\label{def:negativeAssociation}
A set of random variables $X_1, ..., X_n$ is said to be negatively associated if for any two disjoint index sets $i, j \subseteq [n]$ and two functions $f$, $g$ both monotone increasing or both monotone decreasing it holds:
$$
\E[f(X_i: i \in I) \cdot g(X_j: j \in J)] \leq \E[f(X_i: i \in I)] \cdot \E[g(X_j: j \in J)]
$$
\end{definition}

\section{The Algorithm for Selecting $Q$}
In this section, we present a formal statement of the algorithm employed to construct the subgraph $Q$. We then explain how we can use the tools we  provide later in the paper to show  quering $Q$ proves \Cref{thm:Themain} (the main theorem). 

In summary, for a given parameter $t = O_{\epsilon}(1/p)$, we draw $t$ matchings from the same distribution as $\opt$ (the optimal solution) and define $Q$ as the union of these matchings.

\begin{tboxalg2e}{The algorithm for constructing $Q$}
\begin{algorithm}[H]
	\DontPrintSemicolon
	\SetAlgoSkip{bigskip}
	\SetAlgoInsideSkip{}
	
	\label{alg:query}
        $Q \gets \emptyset$ \\
	\For{$i$ from $1$ to $t$}{
        Let $G_i$ be a random realization of $G$ containing each edge independently w.p. $p_e$.\\ 
        Set $M_i = \mu(G_i)$ \\
            $Q \gets Q \cup M_i$
	}
	
	\Return $Q$
\end{algorithm}
\end{tboxalg2e}
\vspace{2 mm}
Let us define subsets of crucial and noncrucial edges as follows.
\begin{equation}\label{eq:crucialnoncrucial}
    C = \{ e \in E : x_e \geq \tau \} \;\;\;\;\;\text{ and } \;\;\;\;\; N = \{ e \in E : x_e < \tau \},
\end{equation} where 
$\tau=\theta(\frac{1}{t})$ and $t= \frac{1}{\tau \epsilon}$ for a sufficiently small $\epsilon\in (0,1)$. (The actual value of $\tau$ and the other variables used in the paper are presented in Table~\ref{table:values}.)
Note that in the above algorithm, matching $M_1, \dots, M_t$ are independent from each other and come from the same distribution as $\opt$. This means that for any edge $e$ and $i\in [t]$ we have $\Pr[e\in M_i] = Pr[e\in  \opt] = x_e$. As a result, the subgraph outputted by this algorithm contains almost all the crucial edges. Moreover, it  picks any non-crucial edge $e\in  N$ with a large enough probability as a function of their $x_e$.  We formally state these properties below in  Claim~\ref{claim:allcrucial} and Claim~\ref{claim:lowerbound_einQ}. While the proofs are pretty straightforward, we include them in Section~\ref{section:proofs} for the sake of completeness.

\newcommand{\claimallcrucial}[0]{Given constant numbers  $\tau, \epsilon \in (0,1)$, Let $Q$ be the output of Algorithm~\ref{alg:query} with parameter $t\geq \frac{1}{\tau \epsilon}$.  Any crucial edge $e\in C$ with $x_e \geq  \tau$  is present in $Q$ with probability at least $1-\epsilon$. }
\begin{claim}\label{claim:allcrucial}
\claimallcrucial{}
\end{claim}

\newcommand{\claimlowerboundeinQ}[0]{Any edge  $e\in  E$ is present in $Q$ with probability at least $\min(1/3,  tx_e/3).$}

\begin{claim}
    \label{claim:lowerbound_einQ}
    \claimlowerboundeinQ{}
\end{claim}

\subsection{Proof of the Main Theorem.}  As discussed previously  in  Section~\ref{sec:techniques}, to prove  Theorem~\ref{thm:Themain}, we will show that $\mc{Q}$ contains a $0.681$-approximate matching. Since $\mc{Q}$ is Union of $t=O_\epsilon(1/p)$ matchings, this proves our main result. In \Cref{def:variance-bounding} we define {\em variance-bounding} matching  algorithms. In Lemma~\ref{lem:reduction}, we prove that for any $\alpha\in (0,1)$ and a small enough constant $\epsilon>0$ existence of an $\alpha$-approximate variance-bounding algorithm implies $\mc{Q}$ contains a $(\frac{1}{2-\alpha}-\epsilon)$-approximate matching. In Lemma~\ref{lem:vertexindependent}, we prove the existance of a $\frac{8}{15}$-approxumate variance-bounding matching implying that $\mc{Q}$ contains a matching with an approximattion  ratio of 
$$\frac{1}{2-8/15}-\epsilon
=\frac{15}{22}-\epsilon \geq 0.6818-\epsilon.$$
By picking a small enough $\epsilon\leq 0.0008$ this gives us an approximation ratio of $0.681$.
\section{The Reduction}
In this section, we first introduce {\em variance-bounding} matching algorithms and then show that the
the existence of an $\alpha$-approximation variance-bounding matching  algorithm  implies that it is possible to find  a $(\frac{1}{2-\apx}-\epsilon)$-approximate matching  with $O_\epsilon(1/p)$ queries per vertex. 

\begin{definition}[Variance-bounding (VB) matching algorithm]\label{def:variance-bounding} We call a matching algorithm $\VB$ an \apx-approximation variance-bounding algorithm if it has the following properties. It takes as input (1)  a graph $H=(V, E)$ whose edges are realized independently, each with a given probability $p_e$ forming subgraph $\mathcal{H}$,  and (2) a matching $M_{\mc{O}}$ of $\mathcal{H}$ found by an arbitrary (potentially randomized) algorithm. The algorithm then outputs a  matching $M_c$ of $\mathcal{H}$ and a subset $A$ of vertices that are unmatched in  $M_c$ \footnote{$A$ is just a subset of unmatched vertices, so some unmatched vertices may not be in $A$.} such that: \begin{enumerate}[label=$(\roman*)$]

\item $M_c$ is in expectation  an \apx-approximate matching with respect to $M_{\mc{O}}$.\label{viprop:approx}
\item For any vertex $v \in V$, $\Pr[v\in A] \geq \Pr[v \notin M_{\mc{O}}]$. \label{viprop:Asideprob}

\item \label{viprop:independence} For any two vertices $u, v$ that do not have an edge in $H$ the following holds. $\Pr[\{u, v\} \subset A]\geq \delta$ for a constant $\delta\in (0, 1)$.

\item \label{viprop:variance}  Given a parameter $\tau\in (0,1)$, let $\bm{x}$ be a fractional matching on $\overline{H} = (V, \overline{E})$ (the complement of  $H$) with $x_e\leq \tau$ for any $e\in \overline{E}$. 
For  any vertex $v$ variable $Z_v$, defined below, satisfies  $\var(Z_v)\leq \frac{6\tau}{\delta^2}.$ 
$$Z_v = \sum_{e=(u,v)\in \overline{E}}\frac{x_e}{\Pr[\{u,v\}\subset A]}\mathds{1}_{u\in A}.$$ 

\end{enumerate}
\end{definition}

Let us briefly explain why we need a variance-bounding matching algorithm. We will use this algorithm on all the realized crucial edges (i.e., $H=(V, C)$) and define $M_{\mc{O}}$ to be a matching with the  expected weight the same as the contribution of the crucial edges to the optimal solution. We formally define these  inputs in \Cref{def:input}. This gives us a matching $M_c$ on the crucial edges and a subset $A$ of vertices unmatched in  $M_c$.  We will then construct a fractional matching $\bm{f}$ with a small integrity gap exclusively using the (queried and realized) non-crucial edges between vertices in $A$. We need Property~\ref{viprop:approx} to ensure that $M_c$ is large with respect to the contribution of crucial edges to the optimal solution. Property~\ref{viprop:Asideprob} ensures that each vertex is available in $A$ with a large enough probability for its non-crucial edges to be able to contribute to $\bm{f}$ almost as much as their contribution to $\opt$. We need Property~\ref{viprop:independence} to ensure that each edge is available for potential contribution to $\bm{f}$ with a large enough probability. Finally, we will use Property~\ref{viprop:variance} to prove that constructing $\bm{f}$ in a particular way does not result in fractional degrees of vertices exceeding one too often. For more details about the importance of this property, see Section~\ref{section:expectedweightfrac}.

\begin{lemma}[The Reduction]\label{lem:reduction}
For constant numbers $\apx\in (0.5,1)$ and $\epsilon\in (0,0.1)$, existence of an \apx-approximation variance-bounding algorithm $\VB$ (from \Cref{def:variance-bounding}) implies a $(\frac{1}{2-\apx}-\epsilon)$
approximation algorithm for the weighted stochastic matching problem with $O_\epsilon(1/p)$ queries per vertex. 
\end{lemma}

We will prove that the existence of an \apx-approximation variance-bounding algorithm implies that querying the subgraph $Q$ outputted by Algorithm~\ref{alg:query} with parameter $t=O_\epsilon(1/p)$ gives us a $(\frac{1}{2-\apx}-\epsilon)$-approximate solution. Before formally proving this in Section~\ref{sec:reductionproof}, we need to prove a series of other claims and provide some definitions. Below, we give a brief overview of the proof.

The first step of the proof is using the variance-bounding algorithm on the subgraph of all the crucial edges.  Recall that $\VB$ takes as input (1)  a graph $H$ whose edges are realized independently, each with a given probability $p_e$ forming subgraph $\mathcal{H}$,  and (2) a matching $M_{\mc{O}}$ of $\mathcal{H}$ found by an arbitrary (potentially randomized) algorithm. 
Below, we detail the values assign to these parameters in  our reduction. 

\begin{definition}[$H$ and $M_{\mc{O}}$]\label{def:input}
In our reduction we choose the following values for $\mathcal{H}$ and $M_{\mc{O}}$:
\begin{enumerate}
    \item We set $H$ to be the subgraph of all the crucial edges $C$. In other words, $H=(V, C)$. \label{def:graphh}
    \item We set $M_{\mc{O}} = \MM(\mc{H}\cup \mc{N^\star})\cap \mc{H}$, where $\mc{H} = \mc{G}\cap H$ is the actual realization of all the crucial edges, and $\mc{N^\star}$ is a random hallucination of the non-crucial edges containing each edge independently with probability $p_e$. Note that $M_{\mc{O}}$ is a matching only on the realized crucial edges. \label{def:matchingMo}
\end{enumerate}
\end{definition}

In the remainder of this paper, when referring to a variance-bounding algorithm without specifying the input, we assume that the variables $H$ and $M_{\mc{O}}$ are defined according to Definition~\ref{def:input}. Executing the variance-bounding algorithm $\VB$ with these predefined inputs gives us a matching $M_c$ on the critical edges and a subset of unmatched vertices denoted as $A$. Using Property~\ref{viprop:approx}, we prove (in Claim~\ref{claim:approx-ind}) that $M_c$ is an $\alpha$ approximation with respect to the contribution of the crucial edges to the optimal solution. Since due to Claim~\ref{claim:allcrucial}, $\mc{Q}$ contains any crucial edge with probability at least $1-\epsilon$, this implies that $Q \cup M_c$ weights, in expectation, at least $(1-\epsilon)\apx$ times the contribution of crucial edges to $\opt$. It is important to note that we apply the algorithm $\VB$ to all realized critical edges, not exclusively those within $Q$ (the queried ones). This approach ensures that the output of $\VB$, consisting of $M_c$ and set $A$, is independent of the choice of $Q$, as outlined in Claim~\ref{claim:approx-ind}.

The next step of the reduction is using  the non-crucial edges among vertices in  $A$  to construct a fractional matching $\bm{f}$. In Lemma~\ref{lemma:expweight}, we use properties of  $\VB$ to ensure that the expected contribution of any non-crucial edge to $\bm{f}$ is almost the same as its contribution to the optimal solution. We then use the fact that these edges are non-crucial (hence have small $f_e$s) to prove in Lemma~\ref{lemma:integrhjk} that $\bm{f}$ has almost no integrity gap. Putting these pieces together, we prove that either union of this rounded matching and $M_c$ is an $(\frac{1}{2-\apx}-\epsilon)$ approximate solution, or simply only using the crucial edges in $\mc{Q}$ gives us this approximation.

In the following claim, we prove two basic properties about $M_c$ and set $A$ and their relation to the set of non-crucial edges in $\mc{Q}$.
\begin{claim}\label{claim:approx-ind}
     Let  $M_c$ and $A$ be the outputs of an $\alpha$-approximation variance-bounding algorithm which takes as input 
 the subgraph $H=(V, C)$ and matching $M_{\mc{O}}$ defined in \ref{def:input}. We have the followings for $M_c$ and $A$:
 \begin{enumerate}
     \item The expected weight of matching $M_c$ is at least $\alpha$ times the weight of matching $\opt\cap C$ (i.e., the contribution of the crucial edges to the optimal solution).
     \item Let $Q$ be the subgraph of edges we choose to query. For any non-crucial edge $e\in N$, the event $e\in \mc{Q}$ is independent of both  $M_c$ and $A$.
 \end{enumerate}
\end{claim}

\begin{proof}
    To prove the first item of this claim, we will first show that matching $M_{\mc{O}}$ defined in \ref{def:matchingMo} has the same expected weight as the contribution of the crucial edges to the optimal solution. In other words, $\E[\weight(\opt \cap C)]=\E[M_{\mc{O}}]$. Recall that we have defined $M_{\mc{O}} = \MM(\mc{H}\cup \mc{N^\star})\cap \mc{H}$, where $\mc{H} = \mc{G}\cap H$ is the actual realization of all the crucial edges, and $\mc{N^\star}$ is a random hallucination of the non-crucial edges containing each edge independently with probability $p_e$. This implies that $\mc{H}\cup \mc{N^\star}$ comes from the same distribution as $\mc{G}$ and as result $M_{\mc{O}}$ is drawn from the same distribution as $\opt$. For any crucial edge $e\in C$ this gives us $\Pr[e\in M_{\mc{O}}]= \Pr[e\in \opt \cap C]$ and $\E[\weight(\opt \cap C)]=\E[M_{\mc{O}}]$.  This proves the first part of the claim since due to Definition~\ref{def:variance-bounding}, property~\ref{viprop:approx} we know $M_{c}$ is an $\alpha$ approximation with respect to  $M_{\mc{O}}$. 
    
   To prove the second part of this claim, observe that event $e\in \mc{Q}$ is a function of $Q$ and the realization of non-crucial edges, while $M_{c}$  and $A$ are obtained from running a variance-bounding matching algorithm with inputs $H=(V,C)$ and $M_{\mc{O}} = \MM(\mc{H}\cup \mc{N^\star})\cap H$. Here, $\mc{H}$ is the actual realization of all the crucial edges while $\mc{N^\star}$ is a random hallucination of the non-crucial edges (not the actual realization). Graph $H=(V,C)$ and function $\MM(.)$ are deterministic which means the only randomization in determining  values of $M_{c}$  and $A$ comes from $\mc{H}\cup \mc{N^\star}$.  Since edges of $\mc{G}$ are realized independently, $\mc{H}\cup \mc{N^\star}$ is independent  of the actual realization of the non-crucial edges. It clearly is also independent of the choice of $Q$. This implies that knowing  the outcome of  event $e\in \mc{Q}$ does not change the distribution of  $M_{c}$  and $A$;  hence they are independent. 
\end{proof}

\subsection{A Fractional Matching on the Non-crucial Edges}\label{sec:sec_five}
In this section, we will construct a fractional matching on the non-crucial edges to augment the matching  we get from running a variance-bounding matching on the crucial edges.
Given a variance-bounding algorithm \VB, let $M_c$ and $A$ be the output of $\VB$ with inputs given according to Definition~\ref{def:input}. To begin, let us define a variable $g_e$ for any non-crucial edge as follows:
\begin{equation}
g_e= \frac{x_e}{\Pr[e\in \mc{Q}]\Pr[{u,v}\in A]},
\end{equation}
where $x_e$ is defined as $$x_e= \Pr[e\in \opt].$$ Ideally, for constructing our fractional matching, we would assign a fractional value of $g_e$ to edge $e$ whenever $e\in \mc{Q}$ and both of its endpoints are in $A$. Since these events are independent due to Claim~\ref{claim:approx-ind}, their joint probability is $\Pr[e\in \mc{Q}]\Pr[{u,v}\in A]$. By constructing a fractional matching in this manner, we achieve $\E[f_e]=x_e$ for any edge $e$ and $\E[\bm{f}\cdot\bm{w}] = \E[\weight(\opt)]$.

However, the challenge lies in the fact that constructing $\bm{f}$ in this way may result in it not being a valid fractional matching, as certain vertices may have a fractional degree greater than one. In other words, $\sum_{(u,v)\in N} f_{(u,v)} >1$ may occur for some vertices $v\in V$. To address this issue, we first scale down the fractional values by a small amount. Subsequently, we discard any vertex whose fractional degree still exceeds one. The challenge then becomes demonstrating that this event does not significantly reduce the expected size of the fractional matching. We formally state the algorithm for constructing  a fractional matching on the non-crucial edges in Algorithm~\ref{alg:non-crucial}.

\begin{tboxalg2e}{A fractional matching on the realized non-crucial edges}
\begin{algorithm}[H]\label{alg:fracmatching}
	\DontPrintSemicolon
	\SetAlgoSkip{bigskip}
	\SetAlgoInsideSkip{}
	\label{alg:non-crucial}

Let $M_c$ and $A$ be the outputs of an $\apx$-approximation variance-bounding matching algorithm with inputs given according to Definition~\ref{def:input}. \\ 

    Let $\bm{f}$ be an empty fractional matching on the subgraph of non-crucial edges.\\
    Let $\epsilon\in (0,1)$ be a small given constant (the same as the one used in Table~\ref{table:values}).\\
    Set  $\gamma = (1-\epsilon^2)/(1+ \frac{3\epsilon}{10})$. \\ 
	\For{$e=(u,v)\in N$}{
       Let $g_e=x_e/\Pr[e \in \mc{Q} \text{ and }\{u, v\} \subset A])$\\
       \If{ $e\in \mathcal{Q}$ and both of its end-points are in $A$} {Set $f_e=g_e\gamma$.}\Else{Set $f_e=0$.}
	}
    If the fractional degree of a vertex $v$ in $\bm{f}$ exceeds one (i.e., $\sum_{e\ni v} f_e > 1$), zero out the fractional value of its edges.\\ \label{line:fractionalkill}
	\Return $\bm{f}$
\end{algorithm}
\end{tboxalg2e}
\vspace{4mm}

Since our ultimate goal is to demonstrate the existence of a large weight integral matching on $\mathcal{Q}$ rather than a fractional one, let us first address the integrality gap of the fractional matching produced by this algorithm. We first prove an upper bound of $\epsilon^3$ for $g_e$ of non-crucial edges in Claim~\ref{claim:geupperbound}. We then use this in Lemma~\ref{lemma:integrhjk} to prove that the output of Algorithm~\ref{alg:fracmatching} has a small integrality gap. To help with the flow of the paper, both proofs are deferred to \Cref{section:proofs}.

\newcommand{\claimgeupperbound}[0]{By choosing a sufficiently small constant $\epsilon>0$  in Algorithm~\ref{alg:fracmatching},  we get $g_e \leq \epsilon^3$ for all non-crucial edges.}
\begin{claim}\label{claim:geupperbound}
\claimgeupperbound{}
\end{claim}

\newcommand{\lemmaintegrhjk}[0]{Consider the fractional matching $\bm{f}$ produced by Algorithm~\ref{alg:non-crucial}. There exists an integral matching on the non-crucial edges of $\mc{Q}$ between vertices in $A$ with  weight at least $(1-\epsilon/2)\bm{f}\cdot\bm{w}.$}

\begin{lemma}\label{lemma:integrhjk}
    \lemmaintegrhjk{}
\end{lemma}

\paragraph{Survival of vertices and non-crucial edges} For any vertex $v\in V$, we say $v$ \emph{survives}  Algorithm~\ref{alg:non-crucial}  iff  it is in set $A$, and it is not killed in Line~\ref{line:fractionalkill} of the algorithm (i.e., its fractional degree is not reduced to one). We also say an edge $e$ survives the algorithm iff both its endpoints survive (regardless of whether it is in $\mc{Q}$ or not).

\subsection{Expected Weight of the Fractional Matching}\label{section:expectedweightfrac} Let $\bm{f'}$ denote the value of $\bm{f}$ constructed by Algorithm~\ref{alg:non-crucial} before it zeroes out certain $f_e$ values in Line~\ref{line:fractionalkill}. As discussed earlier in this section, it is evident that $\E[f'_e] = \gamma x_e$ for any edge in $e\in N$. Since $\gamma$  deviates from one by a small constant, the expected weight of $\bm{f'}$ is a sufficiently large approximation relative to the contribution of the non-crucial edges to the optimal solution. Thus, the primary challenge lies in proving that we do not incur a substantial loss by zeroing out certain $f_e$ values in Line~\ref{line:fractionalkill}. Roughly speaking, we only have the opportunity to use an edge $e=(u, v)$ whenever it is in $\mc{Q}$ and its endpoints are in $A$ (i.e., $f'_e\neq 0)$, and we lose this opportunity if at least one of its endpoints does not survive Algorithm~\ref{alg:fracmatching}. That is, we have:

$$\Pr[f_e\neq 0]= \Pr[f'_e\neq 0] - \Pr[u \text{ or } v \text{ does not survive } \mid f'_e\neq 0].$$
To quantify the amount of loss per edge, we need to upper-bound $\Pr[u \text{ or } v \text{ do not survive } \mid f'_e\neq 0]$ and show that it is significantly smaller compared to $\Pr[f'_e\neq 0]$. Note that here, $f'_e\neq 0$ is not independent of $e$'s end-points surviving since it contributes to their fractional degree. Furthermore,  $e\in Q$ is correlated, albeit negatively (see Claim~\ref{lem:NA_Q}), with the existence of its neighboring non-crucial edges in $Q$, which may also impact the fractional degrees of $u$ and $v$ in $\bm{f'}$.  However, it is still helpful to first upper-bound the probability of $u$ and $v$ surviving without conditioning on  $f'_e\neq 0$. The intuition behind this is that since $f'_e$ is very small (i.e., upper-bounded by $\epsilon^3$  due to Claim~\ref{claim:geupperbound}), its impact on the fractional degree of each endpoint is insignificant. Moreover, since $e\in Q$ is negatively associated with $e'\in Q$ for any non-crucial $e'\neq e$ connected to $u$ or $v$,  conditioning on $e\in Q$ does not increase their $f_{e'}'$. To upper-bound $\Pr[v \text{ does not survive}]$ for any vertex $v$, let us define

\begin{equation}
Y_v = \sum_{e=(u,v) \in N} g_e \cdot \mathds{1}_{u\in A} \cdot \mathds{1}_{e\in \mathcal{Q}}.
\end{equation}
Since we set $f'_e=g_e\gamma$ iff $e\in Q$ and $\{u , v\}\subset A$, whenever vertex $v$ is present in $A$ we have $$Y_v/\gamma =\sum_{e=(u,v) \in N} f'_e.$$ Hence, vertex $v$  survive Algorithm~\ref{alg:fracmatching} iff $Y_v/\gamma\leq 1$. In Lemma~\ref{lem:boundedvary}, we prove that random variable $Y_v$ is concentrated around its mean for any vertex. This analysis crucially relies on Property~\ref{viprop:variance} of variance-bounding algorithms. This would have been enough if we knew $\E[Y_v]$ is sufficiently close to one. While we do not exactly have this, we can use the second property of the variance-bounding algorithms to show $\E[Y_v | v\in A]\leq 1$. This is only doable thanks to Property~\ref{viprop:Asideprob}. Combining all these together, we are able to finally prove in Lemma~\ref{lem:yvclose1} that $Y_v$ is sufficiently close to one, with a sufficiently large portability.

Since both Lemma~\ref{lem:boundedvary} and Lemma~\ref{lem:yvclose1}  have lengthy and technical proofs, we respectively allocate Section~\ref{sec:prooflemma} and Section~\ref{sec:yvclosetoone} to present detailed proofs for them. Finally, we put  the pieces together in Lemma~\ref{lemma:expweight} to demonstrate that $\E[f_e]$ is sufficiently large compared to $x_e$ (the contribution of $e$ to the optimal solution).

\newcommand{\besttau}{\ensuremath{\frac{1}{\epsilon t}}\xspace}
\newcommand{\besttQAlg}{\ensuremath{\frac{1}{\delta^2 p_{min} \epsilon^{3}}}\xspace}
\newcommand{\Inf}{\ensuremath{I}\xspace}

\newcommand{\difffrommu}{\ensuremath{\epsilon}\xspace}
\newcommand{\lemboundedy}[0]{
For any vertex  $v\in V$ define random variable $$Y_v = \sum_{e=(u,v) \in N} g_{e} \cdot\mathds{1}_{u\in A}  \mathds{1}_{e\in \mathcal{Q}}, $$
where $g_e=\frac{x_e}{\Pr[e \in \mathcal{Q}]\cdot\Pr[\{u, v\} \subset A]}.$
The following inequality holds for these random variables.

$$\Pr\Big[|Y_v - \mathbb{E}[Y_v]| \geq \eta \Big] \leq \beta$$
 for $\beta= \frac{\epsilon^2}{100}$ and $\eta=\frac{\epsilon}{10}.$
 
}

\begin{lemma}
\label{lem:boundedvary}
\lemboundedy
\end{lemma}

\newcommand{\lemyvaboutone}[0]{
 For any vertex $v\in  V$ we have:
    $$
     \Pr\left[ Y_v\geq 1+3\eta\right]\leq \beta
    $$
}

\begin{lemma}
\label{lem:yvclose1}
\lemyvaboutone
\end{lemma}

\begin{definition}
    For a vertex $u'$ we define $Y_v(-u')$ to be the summation that we have for $Y_v$ except for the edge $e' = (u', v)$. Formally:
    $$Y_v = \sum_{e=(u,v) \in N, u \neq u'} g_{e} \cdot\mathds{1}_{u\in A}  \mathds{1}_{e\in \mathcal{Q}}, $$
\end{definition}
\newcommand{\Aset}{\ensuremath{\bm{A}}\xspace}
\newcommand{\qstar}{\ensuremath{\mathcal{Q}}\xspace}
\newcommand{\rva}{\ensuremath{S_A}\xspace}
\newcommand{\rvq}{\ensuremath{S_Q}\xspace}
\newcommand{\Cov}{\ensuremath{Cov}\xspace}

\begin{lemma}\label{lem:condeinQlowerbound}
    For every edge $e' = (v, u')$ and constant $\lambda \in (0,1)$ we have:
    $$\Pr[Y_v(-u') > \lambda] \geq \Pr[Y_v(-u') > \lambda \mid e' \in \mathcal{Q}]$$
\end{lemma}

\begin{proof}
    To prove this, let us look at the random variables that affect $Y_v(-u')$. One collection is the set of variables for vertices being present in $A$, i.e., $\rva = \{\mathds{1}_{u \in A} : e = (u, v) \in N  \} $ and the second collection is the edges being present in $\qstar$, i.e. $\rvq = \{\mathds{1}_{e \in \qstar} : e = (u, v) \in N \}$.  Let us fix the randomization on the set $\rva$ and show that for any arbitrary realization of  $A$ represented by $\Aset$ we have: 
    \begin{equation}\label{eq:condA}
      \Pr[Y_v(-u') > \lambda \mid \Aset] \geq \Pr[Y_v(-u') > \lambda \mid e' \in \mathcal{Q}, \Aset]
    \end{equation}
    We prove in ~\cref{lem:NA_Q} that random variables in $\rvq$ are negatively associated. Furthermore, due to Claim~\ref{claim:approx-ind}  they are independent from random variables in $\rva$. So fixing $\Aset$ does not change the distribution of variables in $\rvq$. Now, we use the definition of negatively associated variables. In \cref{def:negativeAssociation} we set the function $f$ to be $\mathds{1}_{e \in Q}$ and $g$ to be $\mathds{1}_{Y_v(u') > \lambda | \Aset}$. Both functions are monotonically increasing on their inputs. So. we  conclude
    $$\E[f \cdot g] \leq \E[f] \cdot \E[g]$$
    This means
    $$\Pr[f = 1] \cdot 1 \cdot \E[g \mid f = 1] + \Pr[f = 0] \cdot 0 \cdot \E[g|f = 0] \leq \Pr[f = 1] \cdot \Pr[g = 1]$$
    The right-hand side has a zero term and by canceling out the term $\Pr[f = 1]$ from both sides we get 
    $\E[g \mid f = 1] \leq \Pr[g = 1]$ which proves \eqref{eq:condA}. Now, going over all possible realizations for $\Aset$ we get
    \begin{align*}
        \Pr[Y_v(-u') > \lambda] &= \sum_{\Aset}  \Pr[Y_v(-u') > \lambda \mid \Aset] \cdot \Pr[A = \Aset] \\
        & \geq \sum_{\Aset} \Pr[Y_v(-u') > \lambda \mid e' \in \mathcal{Q}, \Aset] \cdot \Pr[A = \Aset] \\
        & \geq \Pr[Y_v(-u') > \lambda \mid e' \in \mathcal{Q}]
    \end{align*} \qedhere
\end{proof}

\begin{lemma}\label{lemma:expweight}
For every non-crucial edge $e = (u, v)$ we have $\E[f_e] \geq (1 - \epsilon/2) \cdot x_e.$   
\end{lemma}

\begin{proof}
For an edge $e = (u, v)$ let us see what the contribution of this edge to the fractional matching is. The value of $f_e$ in the output of Algorithm~\ref{alg:fracmatching} will be $g_e \cdot \gamma$ if we have $e \in \mathcal{Q}$, $u$ and $v$ are in $A$, and $f_e$ is not zeroed out in Line~\ref{line:fractionalkill} due to the fractional degree of $u$ or $v$ exceeding one. Let us define the event $E_u$ being $\gamma \cdot Y_u(-v) \leq 1 - g_e$ and event $E_v$ to be $\gamma \cdot Y_v(-u) \leq 1 - g_e$. If we have $E_u = 1$ then the fractional degree of $u$ will not exceed 1 and similarity if we have $E_v = 1$ then the fractional degree of $v$ will not exceed 1. Hence we have:
\newcommand{\Overline}[1]{\hspace{0.2em}\overline{#1}\hspace{0.2em}}
\begin{align}\label{eq:ExpectedFe}
    \E[f_e] & \geq g_e \cdot \Pr[\{u, v\} \subset A \And E_u \And E_v \And e \in \mathcal{Q}]
    \nonumber
    \\
    & \geq g_e \cdot \Pr[e \in \mathcal{Q}] \cdot \Pr[\{u, v\} \subset A \And E_u \And E_v \mid e \in \mathcal{Q}] 
    \nonumber
    \\
    & \geq g_e \cdot \Pr[e \in \mathcal{Q}] \cdot (\Pr[\{u, v\} \subset A \mid e \in \mathcal{Q}] - \Pr[\Overline{E_u} \mid \Pr[e \in \mathcal{Q}]] - \Pr[\Overline{E_v} \mid \Pr[e \in \mathcal{Q}]]) 
    \nonumber
    \\
    & \geq g_e \cdot \Pr[e \in \mathcal{Q}] \cdot (\Pr[\{u, v\} \subset A \mid e \in \mathcal{Q}] - \Pr[\Overline{E_u}] - \Pr[\Overline{E_v}]) 
    \nonumber
    \\
    & \geq g_e \cdot \Pr[e \in \mathcal{Q}] \cdot (\Pr[\{u, v\} \subset A] - \Pr[\Overline{E_u}] - \Pr[\Overline{E_v}])
\end{align}
The second to last step is applying \cref{lem:condeinQlowerbound} with $\lambda = \frac{1 - g_e}{\gamma}$ for the edge $e$ on both sides. The last step is due to the set $A$ being independent of the subgraph $\mathcal{Q}$ as a result of Claim~\ref{claim:approx-ind}.

To continue, we need upper-bounds for $\Pr[\Overline{E_u}]$ and  $\Pr[\Overline{E_v}]$. Due to Lemma~\ref{lem:yvclose1} we have $$\Pr\Big[\frac{Y_v}{1 + 3 \eta} \geq 1 \Big] \leq \beta,$$
Which implies
\begin{equation}\label{eq:num13}
\Pr\Big[\frac{Y_v}{1 + 3 \eta} \cdot (1 - g_e) \geq 1 - g_e \Big] \leq \beta.
\end{equation}
Hence we get
\begin{align*}
    \Pr[\Overline{E_u}] &\leq \Pr[\gamma \cdot Y_u(-v) \geq 1 - g_e] && \\
    &\leq \Pr[\gamma \cdot Y_u \geq 1 - g_e] \\
    &= \Pr[\frac{1 - \epsilon^2}{1 + 3 \eta} \cdot Y_u \geq 1 - g_e] \\
    & \leq \Pr[\frac{1 - g_e}{1 + 3 \eta} \cdot Y_u \geq 1 - g_e] && 
    \text{(due to \cref{claim:geupperbound} we have $g_e \leq \epsilon^2$)} \\
    & \leq \beta &&
    \text{(due to Equation~\ref{eq:num13})}\\
\end{align*}
Combining this with \eqref{eq:ExpectedFe} we get:
\begin{align*}
    \E[f_e] &\geq g_e \cdot \Pr[e \in \mathcal{Q}] \cdot (\Pr[\{u, v\} \subset A] - \Pr[\Overline{E_u}] - \Pr[\Overline{E_v}]) \\
    & \geq \frac{x_e}{\Pr[e \in \mathcal{Q}]\cdot\Pr[\{u, v\} \subset A]} \cdot \Pr[e \in \mathcal{Q}] \cdot (\Pr[\{u, v\} \subset A] - \Pr[\Overline{E_u}] - \Pr[\Overline{E_v}]) \\
    & \geq \frac{x_e}{\cdot\Pr[\{u, v\} \subset A]} \cdot  (\Pr[\{u, v\} \subset A] - \Pr[\Overline{E_u}] - \Pr[\Overline{E_v}]) \\
    & \geq \frac{x_e}{\cdot\Pr[\{u, v\} \subset A]} \cdot  (\Pr[\{u, v\} \subset A] - 2 \beta) \\
    & \geq x_e \cdot (1 - \frac{2\beta}{\Pr[\{u, v\} \subset A]}) \\
    & \geq x_e \cdot (1 - \frac{\epsilon^2}{50\delta})
\end{align*}
For a sufficiently small (i.e., $\epsilon\leq 25\delta$), we get $\frac{\epsilon^2}{50\delta} \leq \epsilon/2$ completing the proof. \qedhere
\end{proof}

\subsection{Proof of Lemma~\ref{lem:reduction} (The Reduction)}\label{sec:reductionproof}
In this section, we will put all the pieces together to formally prove Lemma~\ref{lem:reduction}. Let $Q$ be the subgraph outputted by Algorithm~\ref{alg:query}. 
We prove that the existence of an  $\alpha$-approximation variance-bounding matching algorithm means $\mc{Q}$, the realization of $Q$, contains a $\frac{1}{2-\alpha}-\epsilon$ approximate solution. Since $Q$ is the union of $t=\frac{1}{\tau \epsilon}$ matchings, plugging in the value of $\tau={p\epsilon^5 \delta^2}$ from Table~\ref{table:values} implies $Q$ has max-degree $O_{\epsilon}(1/p)$. Therefore proving that $\mc{Q}$ contains a $\frac{1}{2-\alpha}-\epsilon$ approximate solution proves this lemma.

Let $M_c$ and $A$ be the outputs of the $\alpha$-approximation variance-bounding algorithm on inputs specified in Definition~\ref{def:input}. Recall that $M_c$ is a matching on the crucial edges and $A$ is a subset of vertices unmatched in $M_c$. Let $\sigma$ be the ratio of the optimal solution that comes from the crucial edges. That is
$$\sigma=\frac{\sum_{e\in C} \Pr[e\in \opt] w_e}{\weight(\opt) }.$$
Due to Claim~\ref{claim:approx-ind}, we know that the expected weight of $M_c$ is $\alpha\sigma$ fraction of the optimal solution. Furthermore, we showed in Claim~\ref{claim:allcrucial} that any crucial edge belongs to $Q$ with probability at least $(1-\epsilon)$. As a result we have \begin{equation}\label{eq:nierfjerf}
    \E[\weight(M_c\cap \mc{Q})]\geq (1-\epsilon)\alpha\sigma\weight(\opt).
\end{equation}

The next step is to use the non-crucial edges among vertices in $A$ to augment $M_c\cap \mc{Q}$. In Lemma~\ref{lemma:expweight}, we prove that it is possible to construct a fractional matching $\bm{f}$ on the non-crucial edges among vertices in $A$ such that for any non-crucial edge $$\E[f_e]\geq (1-\epsilon/2)\Pr[e\in \opt].$$ Hence, $\E[\bm{f}\bm{w}]\geq (1-\sigma) (1-\epsilon/2)\weight(\opt).$ As a result of Lemma~\ref{lemma:integrhjk} it is possible to round $\bm{f}$  and achieve an integral matching $M_n$ such that
\begin{equation}\label{eq:koqwpmvgy3}
   \E[\weight(M_n)]\geq (1-\epsilon/2)(1-\sigma) (1-\epsilon/2)\weight(\opt)\geq (1-\epsilon)(1-\sigma)\weight(\opt).
\end{equation}
Putting \cref{eq:nierfjerf} and \cref{eq:koqwpmvgy3} together implies the existence of a matching in $\mc{Q}$ with expected weight at least
$$(1-\epsilon)\alpha\sigma\weight(\opt)+ (1-\epsilon)(1-\sigma)\weight(\opt)  = (1-\epsilon)\weight(\opt)(1-\sigma+\alpha\sigma).$$

We claim that the best of this matching  and simply taking the max-weight matching among the crucial edges of $\mc{Q}$ gives us the desired approximating ratio. Since each crucial edge belongs to $Q$ w.p. at least $1-\epsilon$, its realization contains a matching with expected weight at least $(1-\epsilon)$ times the contribution of crucial edges to the optimal solution which is $(1-\epsilon)\sigma\weight(\opt).$ The best of these two solutions gives us the approximation ratio of at least $$(1-\epsilon)\cdot\max(\sigma, 1-\sigma+\alpha\sigma)\geq (1-\epsilon)\frac{1}{2-\alpha}\geq \frac{1}{2-\alpha}-\epsilon.$$
Hence, this implies that the realization of subgraph $Q$ with max-degree $O_{\epsilon}(1/p)$ contains a $(\frac{1}{2-\alpha}-\epsilon)$- approximate solution completing the proof of Lemma~\ref{lem:reduction}.

\section{A Variance-Bounding Matching Algorithm}\label{secvariancebounding}
In this section, we discuss the existence of an $\frac{8}{15}$-approximation variance-bounding matching algorithm.
\begin{lemma}[Variance-bounding Matching Lemma]\label{lem:vertexindependent}
There exists an $8/15$-approximation variance-bounding matching algorithm (as defined in \Cref{def:variance-bounding}).
\end{lemma}

To prove this lemma, we will design  an algorithm which achieves the properties of a variance bounding matching algorithm discussed in \Cref{def:variance-bounding}.
One of the main tools we employ in our design  is an algorithm designed for online matching due to Fu, Tang, Wu, Wu, and Zhang ~\cite{Fu}. We will refer wto this algorithm as FTWWZ (the authors' initials). Below, we briefly review the main components of their work to the extent required for presenting our algorithm and results.

\paragraph{Batched RCRS:} \label{sec:Fualgorithm} In our algorithm, we will use a batched  random order contention
resolution scheme (RCRS) in the FTWWZ algorithm for the following online matching problem. We are given a graph $G=(V,E)$ along with a fractional matching $\bm{y}$ on the graph. The graph is revealed in an online manner as follows.  Vertices arrive in a uniformly random order given by a permutation $\pi$. Upon the arrival of a vertex $v$, the status of the edges connecting $v$ to vertices before it (i.e., all vertices $u$ that $\pi_u < \pi_v$)  are revealed, namely a \emph{batch} of edges. 
Then at most one of the edges in the batch becomes \emph{active} such that $$\Pr[e \text{ becomes active}] =  y_e$$ for any edge $e$ in  the batch. A batched RCRS decides, upon the arrival of each vertex, irrevocably whether to select the active edge (if any exists). At any point in time, the selected edges must form a matching. Given a parameter $\alpha$ a batched RCRS is called $\alpha$-selectable if it picks each active element with probability at least $\alpha$. FTWZZ provides a simple greedy algorithm which is $\frac{8}{15}$-selectable. The algorithm starts by  modifying the graph and then greedily matches any active edge if its endpoints are unmatched. One of these modifications is lowering the active probability of each edge using a function \begin{equation}\label{g(yv)}
    g(y_e)=\frac{3y_e}{3+2y_e}.
\end{equation} Below is the main result we use from FTWZZ's work.

\begin{proposition}\label{prop:fu}
   If in the above-mentioned setting, each edge becomes active with probability $g(y_e)=\frac{3y_e}{3+2y_e}$ then it is possible to design an RCRS for constructing a matching which selects each edge w.p. at least $\frac{8}{15}y_e$.
\end{proposition}

We state our variance-bounding matching algorithm formally in Algorithm~\ref{alg:Fu-modified}. The  algorithm starts by drawing a  random permutation $\pi$ over the vertices uniformly at random. We then let the vertices arrive in the order given by this permutation. Upon arrival of a vertex $v$, we look at the realization of its edges to the vertices with smaller $\pi_u$. Then, a random process decides which one of its edges (if any) becomes active. We explain  this process in \Cref{def:Fuv}. The process is designed in a way that the probability of each edge becoming active is $g(\Pr[e\in M_{\mc{O}}])$ where $M_{\mc{O}}$ is the random matching in the statement of Lemma~\ref{lem:vertexindependent}. Proposition~\ref{prop:fu} implies that using  FTWZZ's RCRS we can  construct a matching $M_c$ which $\frac{8}{15}$-approximates  $M_{\mc{O}}$. We then let set $A$ be the set of all the vertices that do not have an active edge throughout the algorithm. Note that any vertex in $A$ is also unmatched since only active edges can join  the matching. However, there may be some unmatched vertices that are not in $A$.

\begin{definition}(Edge Activation Process)\label{def:Fuv} The activation probability of the edges in this process comes from matching $M_{\mc{O}}$  in the statement of 
Definition~\ref{def:variance-bounding}. Recall that $M_{\mc{O}}$ is a matching on the realized crucial edges. See Definition~\ref{def:input} for what we set $M_{\mc{O}}$ to.

First, let us define 
\begin{equation}\label{eq:yvjkrgjrg}
    y_e= \Pr[e\in M_{\mc{O}}].
\end{equation}
Here, the randomization comes from the realization of edges in $\mathcal{H}$ and the algorithm for finding $M_\mc{O}$. Note that $\bm{y}$ is a fractional matching since each vertex  joins $M_{\mc{O}}$ with probability at most one. Moreover, define set $E_{v, \pi} = \{u \in E : \pi_u < \pi_v\}$ to be all of $v$'s  edges to vertices $u$ with $\pi_u < \pi_v$. After looking  at the realization of all these edges let $y'_e$ be the probability of $e$ being in $M_{\mc{O}}$ conditioned on the realization of $E_{v, \pi}$. That is 
\begin{equation}\label{eq:kjrkjfr}
    y'_e= \Pr[e\in M_{\mc{O}} \mid  \text{realization of edges  in } E_{v, \pi}].
\end{equation}
We activate at most one of the realized edges at random such that the probability of any realized edge $e$ being the active one is  
 $\frac{3y'_e}{3 + 2y_e}$. This is possible since $y'_e$ of the realized edges sums up to at most one.
\end{definition}

\newcommand{\invperm}{\ensuremath{{\sigma}^{-1}}\xspace}

We now have all the required tools to design our algorithm (stated below) which we claim satisfies the properties stated in Lemma~\ref{lem:vertexindependent}. 

\begin{tboxalg2e}{Variance-bounding  Matching on $H=(V, E)$}
\begin{algorithm}[H]\label{alg:Fu-modified}
        
	\DontPrintSemicolon
	\SetAlgoSkip{bigskip}
	\SetAlgoInsideSkip{}
	$M_c \gets \emptyset$ and  $A \gets \emptyset$.\\
	Let $\pi$ be a permutation over vertices in $V$ picked uniformly at random.\\
	\For{$v\in V$ in the order of $\pi$}{
           
            Let $E_{v, \pi} = \{u\in E : \pi_u < \pi_v\}$ be all of $v$'s  edges to vertices $u$ with $\pi_u < \pi_v$.\\
            At most one edges $e\in  E_{v, \pi}$ becomes active as described by the edge activation process~\ref{def:Fuv} according to $M_{\mc{O}}$.  \\
		\If{$e=(v, u)$ becomes active and $u$ is unmatched in $M_c$}{  
                FTWZZ's  RCRS from Propostion~\ref{prop:fu} decides whether $e$ joins matching $M_c$ or not.  
            }
	}
        \emph{Kill} all the vertices who had at least one active edge at some point in the algorithm.
        \\
        Let $A$ be the set of remaining alive vertices.
        \\
	\Return $M_c$ and $A$.
\end{algorithm}
\end{tboxalg2e}

\vspace{ 4mm}

\begin{claim}\label{obs:activelessthanx_e}
  For any permutation $\pi$ in Algorithm~\ref{alg:Fu-modified}, the probability of any edge $e$ becoming active is  $g(\Pr[e\in M_{\mc{O}}])$.
\end{claim}
\begin{proof} Fix a vartex $v\in  A$, and let $\mc{E}$ be the realization of edges in  $E_{v, \pi}$ when $v$ arrives. Recall $y_v$ and $y'_v$ which are defined as
$$y_v = \Pr[e\in M_{\mc{O}}] \;\;\;\;\; \text{ and } \;\;\;\;\; y'_v = \Pr[e\in M_{\mc{O}}\mid \mc{E}].$$
Due to the edge activation process we have $$\Pr[e \text{ becomes active}]= \sum_{\mc{E}} \frac{3y'_e}{3+2y_e}\Pr[\mc{E}] = \sum_{\mc{E}} \frac{3 \Pr[e\in M_{\mc{O}}\mid \mc{E}]\Pr[\mc{E}]}{3+2\Pr[e\in M_{\mc{O}}]} = \frac{3 \Pr[e\in M_{\mc{O}}]}{3+2\Pr[e\in M_{\mc{O}}]}.$$ Since this equals $g(\Pr[e\in M_{\mc{O}}])$, the proof of this claim is complete.
\end{proof}

\paragraph{Proof of Property \ref{viprop:approx}.} Consider matching $M_c$  outputted by  Algorithm~\ref{alg:Fu-modified}.  The first property of Lemma~\ref{lem:vertexindependent}  requires $M_c$ to be a $\frac{8}{15}$ approximate matching with  respect to $M_{\mc{O}}$. This is an immediate corollary of Proposition~\ref{prop:fu} 
since we proved in Claim~\ref{obs:activelessthanx_e} that each edge will be activated with probability $g(\Pr[e\in M_{\mc{O}}])$, and Proposition~\ref{prop:fu} states that in this case, FTWZZ's RCRS selects any  edge with probability $\frac{8\Pr[e\in M_{\mc{O}}]}{15}$. Hence using this RCRS in the algorithm results in satisfying Property \ref{viprop:approx}.

\paragraph{Proof of Property \ref{viprop:Asideprob}.} Consider subset $A$ outputted by  Algorithm~\ref{alg:Fu-modified}. The second property requires $\Pr[v\in A] \geq \Pr[v \notin M_{\mc{O}}]$. This claim also holds due to Claim~\ref{obs:activelessthanx_e}. A vertex $v$  does not  joins set $A$ iff it has an active edge at some point in the algorithm. This happens with probability at most $\Pr[v\in M_{\mc{O}}]$ as a result $v$ joins $A$ with probability at least $\Pr[v \notin M_{\mc{O}}]$. Hence, we proved Algorithm~\ref{alg:Fu-modified} satisfies the first two properties.

\paragraph{Proof of Property \ref{viprop:independence}.} The third property requires any two vertices $u$ and $v$ that do not share an edge in the input graph to be in $A$ at the same time w.p. at least an absolute constant $\delta\geq  0$. Due to proof of this property being lengthy we provide it in Claim~\ref{claim:bothalive}.

Before discussing the final property, we will define a set of random variables that together can determine the value of $Z_v$ and importantly are independent from each other. 

\begin{definition}[Influential Random Variables $\{X_1, \dots X_n\}$]\label{def:randomvaljj}
    Consider a run of Algorithm~\ref{alg:Fu-modified} with a fixed permutation $\pi$. For any vertex $u$ we define random variable $X_u$ as follows. If upon arrival of $u$, it has an active edge $e=(u,w)$ then $X_u=w$. Otherwise $X_u=\emptyset$.
\end{definition}

Note that knowing the influential random variables uniquely determines set $A$, as they collectively contain the information regarding which edges become active during the algorithm. A vertex joins set $A$ if it has no active edges during the algorithm. Therefore, by knowing the influential random variables, we have complete knowledge of set $A$. Importantly, these variables exhibit a crucial feature: they are independent of each other. This independence plays an important role in the analysis of Property~\ref{viprop:variance}. We provide the proof of this independence in the claim below.

\begin{claim}
    \label{obs:Fu_properties} 
 The influential random variables  $X_1, \dots, X_n$ defined according to Definition~\ref{def:randomvaljj} for vertices in $V$ are independent.
\end{claim}
\begin{proof}
Let us review which random variables  determine the value of $X_u$ when a vertex $u$ arrives. First, the algorithm looks at the realization of edges in $E_{v, \pi}$ to determine the probability with which each edge  becomes active. Then, the edge activation process decides which edge becomes active. Observe that for a fixed permutation subsets $E_{u, \pi}$ and  $E_{u', \pi}$  are disjoint for any two vertices $u$ and $u'$ (i.e., $E_{u', \pi}$ and $E_{u', \pi}$ do not share any edges). Since the edge realizations are independent, the realization of edges in $E_{u, \pi}$ and  $E_{u', \pi}$ are independent as well.  Furthermore, We can assume that the edge activation process draws fresh random bits every time it is called. This implies that $X_1, \dots, X_n$ are functions of independent sets of random variables hence they are independent. 
\end{proof}

\paragraph{Proof of Property \ref{viprop:variance}.} We formally prove this property in Lemma~\ref{lem:BoundedVarZv}. However, let us give some intuition before diving into it. Let us first investigate what bound we could get for the variance of $Z_v$ if the membership of vertices in $A$ were independent from each other. In that case,  $Z_v$ would be the sum of independent random variables and its variance would  be  the sum of their variance. Therefore, we would get
\begin{align*}
    \var(Z_v) &= \sum_{(u,v)\in \overline{E}}\var\left (\frac{x_{(u,v)}}{\Pr[\{u,v\}\subset A]}\mathds{1}_{u\in A}\right) \\ & \leq \sum_{(u,v)\in \overline{E}} \Pr[u\in A] \left( \frac{x_{(u,v)}}{\Pr[\{u,v\}\subset A]} \right)^2 \\ &  \leq 
    \frac{\sum_{(u,v)\in \overline{E}}x_{(u,v)}^2}{\delta^2}
      \leq 
    \frac{1}{\tau}\frac{\tau^2}{\delta^2} \leq \frac{\tau}{\delta^2}.
\end{align*}
This implies that our desired upper-bound for  $\var(Z_v)$ is just a constant times the upper-bound we would get under independence which is quite a strong claim. To prove this  claim  we first write $Z_v$ as a function of the influential random variables. Since they are independent, we can use the  Efron–Stein inequality to give an upper bound for $\var(Z_v)$. This inequality basically requires to show that redrawing any of the influential random variables changes the  value of $Z_v$ by a small amount in expectation. One of the main reasons we can achieve this  is that redrawing any  $X_u$ can only change the membership of three vertices in $A$. Furthermore, for any pair of vertices $u$ and $w$ we have $\Pr[X_u=w] \leq z'_v$.

\begin{lemma} [Property \ref{viprop:variance}]\label{lem:BoundedVarZv}
Consider set $A$ outputted by  Algorithm~\ref{alg:Fu-modified}. Given any fractional matching $\bm{x}$ on $\overline{H}$ (the complement of  $H$), for any vertex  $v\in V$ define  $$Z_v = \sum_{(u,v) \in  \overline{E}} h_{(u,v)}\mathds{1}_{u\in A},$$ where $h_{e}=\frac{x_e}{\Pr[\{u,v\}\subset A]}.$ If for a parameter $\tau\in (0,1)$, the fractional matching satisfies $x_e<\tau$, then $\var(Z_v)\leq \frac{10\tau}{\delta^2}.$
\end{lemma}

\newcommand{\gstar}{\ensuremath{\mathcal{G}}\xspace}

\begin{proof}
We are able to prove this lemma for any permutation $\pi$. Therefore, in our proof, we will assume that $\pi$ is a fixed arbitrary permutation.

To demonstrate the desired upper bound for $\var(Z_v)$, we will employ the Efron–Stein inequality~\ref{prop:efron-stein}. However, for the inequality to be applicable, we need to express $Z_v$ as a function of independent random variables.  Due to the potential correlation between the presence of different vertices in $A$, it is necessary to identify another set of  random variables which are independent and can be used to determine the value of $Z_v$. This is where the influential random variables defined in Definition~\ref{def:randomvaljj} come into play.

Note that the randomization in $Z$ is derived from the variables $\mathds{1}_{u\in A}$ for all $u\in V$, and the values of these random variables can be determined based on the influential random variables. Consequently, we can express $Z_v$ as a function of the random variables in $\bm{X}$. Let us denote this function as $\funcZ$. Thus, we have the relationship: 
$$Z_v =  \funcZ(\bm{X}).$$ Additionally, we will use $\mathds{1}_{u\in A}(\bm{X})$ to refer to the value of $\mathds{1}_{u\in A}$ when the active edges are determined by the influential random variables $\bm{X}$.
 
Now, we proceed to bound the variance of $Z_v$ using the Efron-Stein inequality. For each vertex $u\in V$, we define a random variable $X'_u$ drawn from the same distribution as $X_u$. Importantly, these random variables are independent both from each other and from $\bm{X}$. Let $\bm{X}^{(u)}$ denote $\bm{X}$ with $X_u$ replaced by $X'_u$. In other words, $\bm{X}^{(u)}$ is identical to $\bm{X}$ except that it contains a new draw of $X_u$, represented by $X'_u$, instead of the original variable $X_u$. By applying the Efron-Stein inequality (provided in Proposition~\ref{prop:efron-stein}), we obtain the following:
    $$ 
    \var(\funcZ(\bm{X})) \leq \frac{1}{2} \sum_{u\in V} \E\left[\left(\funcZ(\bm{X}) - \funcZ(\bm{X}^{(u)}) \right)^2\right].
    $$
We take an arbitrary vertex $u\in V$ and focus on giving  an upper-bound for:
\begin{equation}\label{eq:Lkjhwfjher}
    \E\left[\left(F(\bm{X}) - F(\bm{X}^{(u)}) \right)^2\right].
\end{equation} 
Let us first define a notation. For any vertex $v$, let  $\adj_v$ be the set of neighbors of $v$ in graph $\bar{H}$. Formally: \begin{equation}
    \adj_v = \{u : (u, v) \in \bar{H}\}.
\end{equation} 
If  we have $X_u = X'_u$ then \eqref{eq:Lkjhwfjher} equals zero. Therefore, we only consider the case $X_u \neq X'_u$ in our upper bound. 
Let us assume  $X_u= w$ and $X'_u = w'$ with $w'\neq w$. Since some vertices may not have active edges, we may have $w=\emptyset$, in which case we assume $h_{v,w }=0$ (similarly for $w'$). Since $\bm{X}$ and  $\bm{X}^{(u)}$ only differ in the active edges between vertices $\{u,w, w'\}$, for any vertex $w'' \notin \{u,w, w'\}$ we have $$\mathds{1}_{w''\in A}(\bm{X})=\mathds{1}_{w''\in A}(\bm{X}^{(u)}),$$ which implies
\begin{align*}
    |F(\bm{X})-F(\bm{X}^{(u)})|& \leq h_{(w,v)} + h_{(w',v)} +  h_{(u,v)}.
\end{align*} 
Therefore, the contribution of this event to the total expectation in \eqref{eq:Lkjhwfjher} is upper-bounded by
\begin{align*}
    \Pr\left[X_u=w, X'_u=w'\right]\left(F(\bm{X})-F(\bm{X}^{(u)})\right)^2 & \leq  \Pr[X_u=w]\Pr[X'_u=w']\left(h_{(w,v)} + h_{(w',v)} + h_{(u,v)}\right)^2\\
& \leq  y_{u,w}\cdot y_{u,w'}\left(h_{(w,v)} + h_{(w',v)}+h_{(u,v)}\right)^2\\
& \leq 3\cdot y_{u,w}\cdot y_{u,w'} \left(h_{(w,v)}^2 + h_{(w',v)}^2+h_{(u,v)}^2\right)
\end{align*}
\\
Let us emphasize that if any of the vertices in $\{u, w, w'\}$ is not in $\adj_v$, then in the above term its corresponding  $h$ equals zero. 
Due to Claim~\ref{obs:activelessthanx_e} we know  each edge becomes active w.p. at most $y_e$ where $\bm{y}$ (defined in \ref{eq:yvjkrgjrg}) satisfies the properties of a fractional matching. As such

\begin{align*}
\Pr[X_u=w]\Pr[X'_u=w']\left(h_{(w,v)} + h_{(w',v)} + h_{(u,v)}\right)^2
& \leq  y_{u,w}\cdot y_{u,w'}\left(h_{(w,v)} + h_{(w',v)}+h_{(u,v)}\right)^2\\
& \leq 3\cdot y_{u,w}\cdot y_{u,w'} \left(h_{(w,v)}^2 + h_{(w',v)}^2+h_{(u,v)}^2\right)
\end{align*}
Now, we can give an  upper-bound for $\E\left[\left(F(\bm{X}) - F(\bm{X}^{(u)}) \right)^2\right]$ in  terms of $h_{w,v}^2$ for all $w\in \adj_v$. Note that $h_{w,v}^2$ only appears in the above upper-bounds if $X_u=w$ or $X_u=w'$ which means it gets a coefficient of 
$$2(3y_{u,w}\sum_{w'\in \adj_u}y_{u,w'})\leq 6y_{u,w},$$ where  the last inequality comes from the fact that $\bm{y}$ is a fractional matching. Hence, we have $\sum_{w'\in \adj_u}y_{u,w'}\leq 1$.  For vertex $u$, the coefficient for  $h_{w,u}^2$ is upper-bounded by 
$$\sum_{w\in \adj_u}2(3y_{u,w}\sum_{w'\in \adj_u}y_{u,w'})\leq 6.$$ This implies 
$$\E\left[\left(F(\bm{X}) - F(\bm{X}^{(u)}) \right)^2\right]\leq 6\cdot h_{u,v}^2 + 6\cdot y_{u,w}\cdot h_{(w,v)}^2.$$ 
Putting things together, the Efron–Stein inequality gives us the following bound on variance
\begin{align}
    \var(\funcZ(\bm{X})) \nonumber &\leq \frac{1}{2} \sum_{u\in V} \E\left[\left(\funcZ(\bm{X}) - \funcZ(\bm{X}^{(u)}) \right)^2\right]\\ \nonumber
    &\leq \frac{1}{2} \sum_{u\in V} \left(6 \cdot h_{(u,v)}^2+ 6 \cdot y_{(u,w)} \cdot h_{(w,v)}^2\right)
    \\ \nonumber
    &\leq 3 \sum_{u\in V} h_{(u,v)}^2 +  3 \sum_{w\in V} \sum_{u\in V} y_{(u,w)} \cdot h_{(w,v)}^2
    \\ \nonumber
    &\leq 3 \sum_{u\in V} h_{(u,v)}^2 +  3 \sum_{w\in V} \sum_{u\in V} y_{(u,w)} \cdot h_{(w,v)}^2
    \\ 
    & \leq 6\sum_{w\in \adj_v}h_{(w,v)}^2, 
\end{align} 
where the last sum is only over vertices in $\adj_v$ since for any vertex $w\notin \adj_v$  has $h_{(w,v)}=0$.
Now, let us recall the definition $$h_e = \frac{x_e}{\Pr[\{u, v\} \subset A]}.$$ Due to Property~\ref{viprop:independence}, we have $\Pr[\{u,v\}\in A]\geq \unmapr{}$ and by the statement  of  the lemma, we have $x_e\leq \tau$ which implies
$$ \var(\funcZ(\bm{X}))\leq \frac{6}{\unmapr{}^2}\sum_{w\in \adj_v} x_{(w,v)}^2 \leq \frac{6}{\unmapr{}^2}\frac{1}{\tau}\tau^2 = \frac{6\tau}{\unmapr{}^2},
$$ and completes the proof of this lemma. 
\end{proof}
\section{Proof of Lemma~\ref{lem:boundedvary}}\label{sec:prooflemma}

We devote this section to prove Lemma~\ref{lem:boundedvary} due to it being lengthy and technical.

\restatelem{\ref{lem:boundedvary}}{\lemboundedy{}}
To prove the desired concentration bound on $Y_v$ we begin by bounding its variance. This will allow us to apply Chebyshev inequality (\Cref{prop:Chebyshev}) to prove our desired bound. Let us first examine the random variables that affect $Y_v$'s value.
One collection is the set of variables for presence of vertices after running Algorithm~\ref{alg:Fu-modified} in set $A$, i.e., $\rva = \{\mathds{1}_{u \in A} :  u\in V  \} $ and the second collection is the edges being present in $\qstar$, i.e. $\rvq = \{\mathds{1}_{e \in \qstar} : e = (u, v) \in N \}$. 

By the law of total variance (\Cref{prop:lawof}) we have: 
$$\var[Y_v] = \E[\var(Y_v\mid S_A)] + \var[\E(Y_v \mid S_A)]$$
We will later prove that 

\begin{equation}\label{eq:final_varY}
\E[\var[Y_v | S_A]] \leq 60 \cdot (\epsilon^{6} + \epsilon^{5})
\end{equation}
To bound the term $\var[\E(Y_v\mid S_A)]$ let us first examine what $\E(Y_v\mid S_A)$ is. 

\begin{align*}
\E[Y_v|S_A] & = \sum_{e = (u, v) \in N} \E[g_e \cdot \mathds{1}_{u \in A} \cdot \mathds{1}_{e \in Q} \mid S_A] \\ & = \sum_{e = (u, v) \in N} \E\left[\frac{x_e}{\Pr[e \in \mathcal{Q}] \cdot \Pr[\{u, v\} \subset A]} \cdot \mathds{1}_{u \in A} \cdot \mathds{1}_{e \in Q} \mid S_A\right] \\
& = \sum_{e = (u, v) \in N} \E\left[\frac{\mathds{1}_{e \in Q}}{\Pr[e \in \mathcal{Q}]}\right] \E\left[\frac{x_e}{\Pr[\{u, v\} \subset A]} \cdot \mathds{1}_{u \in A} \mid S_A\right] \\
& = \sum_{e = (u, v) \in N} \E\left[\frac{x_e}{\Pr[\{u, v\} \subset A]} \cdot \mathds{1}_{u \in A}\right]
\end{align*}
To go from the second to the third line, we are using the fact that $A$ and $\mc{Q}$ are independent due to Claim~\ref{claim:approx-ind}.
Note that the term in the last line is $Z_v$ in \Cref{lem:BoundedVarZv}.
Applying  the lemma with the fractional matching $x$ being the edges having $x_e < \tau$ we get that: 
$$\var[\E(Y_v\mid S_A)] \leq \frac{10\tau}{\delta^2}$$
Adding this with what we have from \cref{eq:final_varY} and applying law of total variance we get

\begin{equation}
    \var(Y_v) \leq 60 \cdot (\epsilon^{6} + \epsilon^{5}) + \frac{10\tau}{\delta^2} 
\end{equation}
Since $\tau = 20p\epsilon^5 \delta^2$ by setting $\epsilon$ to be a small enough constant, we can get the bound $\var[Y_v] \leq \frac{\epsilon^4}{10^4}$. This will bound the standard deviation of $Y_v$ by $\frac{\epsilon^2}{100}$ which is used when applying Chebyshev's Inequality (See Proposition~\ref{prop:Chebyshev}) on the  random variable $Y_v$. Now that we have $s \leq \frac{\epsilon^2}{100}$, by applying Chebyshev's Inequality, we get

\begin{equation}\label{eq:usedCheb}
    \Pr\Big[|Y_v - \mathbb{E}[Y_v]| \geq c \cdot s \Big] \leq \frac{1}{c^2}
\end{equation}

Note that we wanted to bound the probability that $Y_v$ deviates from its mean by $\eta$. Now if we have $\eta \geq c \cdot s$, we have
\begin{equation}\label{eq:alphacs}
\Pr\Big[|Y_v - \mathbb{E}[Y_v]| \geq \eta \Big] \leq 
\Pr\Big[|Y_v - \mathbb{E}[Y_v]| \geq c \cdot s \Big]
\end{equation}

By replacing value of $\eta = \frac{\epsilon}{10}$ and the fact that $s \leq \frac{\epsilon^2}{100}$ we can see that it is enough to set $c = \frac{\epsilon}{10}$ to satisfy $\eta \geq c \cdot s$. Therefore by combining \eqref{eq:usedCheb} and \eqref{eq:alphacs} we get

\begin{align*}
\Pr\Big[|Y_v - \mathbb{E}[Y_v]| \geq \eta \Big] & \leq 
\Pr\Big[|Y_v - \mathbb{E}[Y_v]| \geq c \cdot s \Big] \\
& \leq \frac{1}{c^2} \\
& \leq \frac{\epsilon^2}{100} = \beta
\end{align*}

Now that we proved the statement of the lemma using \Cref{eq:final_varY}, we prove it 
which states
$
\E[\var[Y_v | S_A]] \leq 60 \cdot (\epsilon^{6} + \epsilon^{5})
$.
Our first step is to see how random variables in $\rvq$ behave. First of all, random variables in $\rvq$ are not independent since $\qstar$ is a collection of matchings, for two incident edges $e_1$ and $e_2$, when $e_1$ is present in one of the matchings $e_2$ will not be in that matching. This intuition might make us believe that for edges relevant to $\rvq$ because they all intersect at the vertex $v$ their presence in $\qstar$ is pairwise \textbf{negatively correlated}. This is in fact true and for proving it we prove a stronger fact about the random variables which is negative association which implies negative correlation. (see \Cref{def:negativeAssociation} for definition).

\begin{lemma}\label{lem:NA_Q}
Random variables in $\rvq$ are negatively associated. 
\end{lemma}
\begin{proof}

For an edge $e = (u, v) \in N$ let us look at the variable $\mathds{1}_{e \in Q}$. The way we construct $Q$ in Algorithm~\ref{alg:query} is that $\mathds{1}_{e \in Q} = 1$ if the edge $e$ is in the maximum matching of at least one of the graphs $G_i$. Let us define the random variable $\mathds{1}_{e, i}$ to be the indicator that edge $e$ belongs to $M_i$. Then we can write:
$\mathds{1}_{e \in Q} = \max_ {i = 1}^t \mathds{1}_{e, i}$. One observation here is that the random variables for presence of edges in $\rvq$ in a single $G_i$ are NA.\footnote{We use NA instead of negatively associated for brevity} This comes from the fact that at most one of them will be in the maximum matching, i.e. $\sum_{e = (u, v) \in N} \mathds{1}_{e, i} \leq 1$. 
This is because for a set of random variables that their sum is not greater than 1, we know that they are NA (see \cite{dubhashi1996balls}). Therefore, the set of variables, 
$A_i = \{ \mathds{1}_{e, i} : \forall e = (u, v) \in N \}$
they will be NA. 

Take the set of random variables
$
\SX = 
\{
\mathds{1}_{e, i} : \forall i \in [t], e = (u, v) \in N
\}
$. Notice $\SX = \bigcup_{i = 1}^t A_i$. ($A_i$ was previously defined in the above paragraph). We know that the union of independent sets of NA random variables are NA (see \cite{khursheed1981positive}). Therefore, since all $A_1, .., A_t$ are NA and we have $\forall i \neq j: A_i \perp A_j$ we get that random variables in $\SX$ are NA.

Let the subsets $S_i$ in the lemma be the random variables for the presence of a single edge in 
$M_1, ..., M_t$ so we have $S_e = \{ \mathds{1}_{e, i} : \forall i \in [t] \}$. These subsets are disjoint because they are on distinct edges. Also, we let $f_e$ be the maximum of all random variables in $S_e$. We know that concordant monotone functions defined on disjoint subsets of a set of NA random variables are NA (see \cite{khursheed1981positive}). Since maximum is a monotonically increasing function, we can apply this  to get the variables: 
$ \mathds{1}_{e \in Q} = \max_{i = 1}^t \mathds{1}_{e, i}
$ are NA. 
\end{proof}

By definition, negative association implies negative correlation. This means Lemma~\ref{lem:NA_Q} implies that for two edges $e_1 = (u_1, v), e_2 = (u_2, v)$ such that $e_1, e_2 \in N$ we have: 
\begin{equation}\label{eq:negyjhwekj}
    \Cov(\mathds{1}_{e_1 \in \qstar}, \mathds{1}_{e_2 \in \qstar}) \leq 0
\end{equation}

Let us take an arbitrary realization of variables in $\rva$ and call it $\Aset$. Our plan is, given this fixed $\Aset$, first upper-bound $\var[Y_v | \Aset]$. Then, using that, find an upper bound for $\var[Y_v]$. At last, we apply Proposition~\ref{prop:Chebyshev} to prove the statement of the lemma.

Define the random variable $$X_u = (g_{(u, v)} \cdot \mathds{1}_{u \in A} \cdot \mathds{1}_{e \in \qstar} | \Aset).$$ We can see that if $\mathds{1}_{u \in A} = 0$, $X_u$ is always equal to zero, and the inequalities discussed further will be trivial for $\var[X_u]$. In the case that $\mathds{1}_{u \in A} = 1$, $X_u = (g_{(u, v)} \cdot \mathds{1}_{e \in \qstar})$. We can see that $(Y_v | \Aset) = \sum_{(u, v) \in N} X_u$. Now, we are ready to bound the variance of $Y_v$ conditioned on $\Aset$. The first step is to bound the variance of $X_u$: 

\begin{equation}\label{eq:var_1}
\var[X_u] = \var[ g_{e} \cdot \mathds{1}_{u\in A} \cdot \mathds{1}_{e\in \qstar} | \Aset] \leq \var[ g_{e} \cdot \mathds{1}_{e\in \qstar} | \Aset]
\end{equation}
This is because when we have fixed $A$, in the case that $\mathds{1}_{u\in A} = 0$ then variance of $X_u$ is zero and in the case that $\mathds{1}_{u\in A} = 1$ the bound in \Cref{eq:var_1} holds. 

Now we know that $\var[X_u] = \E[X_u^2] - \E[X_u]^2 \leq \E[X_u^2]$ so from \Cref{eq:var_1} we get: 

\begin{equation}\label{eq:var_2}
\var[X_u] \leq \E[X_u^2] \leq \E[(g_e \cdot \mathds{1}_{u\in \qstar} | \Aset)^2] \leq \E[(g_e \cdot \mathds{1}_{u\in \qstar})^2] \leq 
\Pr[e \in \qstar] \cdot g_e^2
\end{equation}
Note that we can remove the condition on $\Aset$ because variables in $\rvq$ and $\rva$ are independent. The last step comes from the fact that with probability $\Pr[e \in \qstar]$, $(g_e \cdot \mathds{1}_{u\in \qstar})^2]$ equals $g_e^2$ and it is zero otherwise. To make further progress, we need a bound on $\Pr[e\in Q]$. The following lemma addresses this.

Expanding $g_e$ in \Cref{eq:var_2}, we get:

\begin{align}\label{eq:beforecasing}
\var[X_u] &\leq
\Pr[e \in \qstar] \cdot 
\left(\frac{x_e}{p_e\cdot\Pr[e \in Q]\cdot\Pr[\{u, v\} \subset A]}\right)^2 \nonumber \\ 
&\leq \frac{x_e^2}{p_e\cdot\Pr[e \in Q]\cdot\left(\Pr[\{u, v\} \subset A]\right)^2} 
\nonumber
\\
&\leq \frac{x_e^2}{p_e\cdot\Pr[e \in Q]\cdot \delta^2} 
\end{align}
To go from the first line to the second, first note the distinction between $Q$ and $\qstar$ in the equation above. By definition of $e \in \qstar$ being $e \in Q \cap e \in \gstar$ we can see that $\Pr[e \in \qstar] = p_e \cdot \Pr[e \in Q]$. This is because $e \in \gstar$ is independent of $e \in Q$ since $Q$ is constructed on hallucinations of $\gstar$. To go from the second line to the third line note that in Lemma~\ref{lem:vertexindependent}, we showed  $\Pr[\{u, v\} \subset A] \geq \delta$.

Moreover, from Claim~\ref{claim:lowerbound_einQ} we know that $\Pr[e\in Q]\geq \min(1/3,  tx_e/3)$ so we consider two cases:

\paragraph{Case 1:} $\Pr[e\in Q]\geq \frac{t \cdot x_e}{3}$.  Combining this and \eqref{eq:beforecasing}  we get:
\begin{align*}
\var[X_u] &\leq \frac{x_e^2}{p_e\cdot\Pr[e \in Q]\cdot \delta^2} \\
&\leq \frac{3x_e^2}{p_e\cdot t \cdot x_e \cdot \delta^2} \\
&\leq \frac{3x_e}{p_e\cdot t \cdot \delta^2}
\end{align*}

\paragraph{Case 2:} $\Pr[e\in Q]\geq \frac{1}{3}$. 
Combining this and \eqref{eq:beforecasing} we get:
$$
\var[X_u] \leq \frac{x_e^2}{p_e\cdot\Pr[e \in Q]\cdot \delta^2}
\leq \frac{3x_e^2}{p_e \cdot \delta^2}
$$
Now that we have a bound on all $\var[X_u]$'s we are ready to bound $\var[Y | \Aset]$. The following proposition is what we need.
\begin{proposition}\label{prop:var_bound}
Let $X$ be a random variable written as the sum of random variables $X_1, ..., X_n$.
So we have $X = \sum_{i = n}^n X_i$. Then we have:

$$\var[X] = \sum_{i = 1}^n \var[X_i] + 2 \cdot \sum_{i = 1}^n \sum_{j > i}^n \Cov (X_i, X_j)$$
\end{proposition}

In \eqref{eq:negyjhwekj} we argued that all variables in $\rvq$ are negatively correlated. Recall the definition of 
$
X_u = g_{(u, v)} \cdot \mathds{1}_{u \in A} \cdot \mathds{1}_{e \in \qstar}.
$
Because we have fixed $\Aset$ all $X_u$'s will be equal to zero or $
g_{(u, v)} \cdot \mathds{1}_{e \in \qstar}.
$ Hence we can argue that $\Cov(X_u, X_w) \leq 0$. This is because if at least one of them is equal to zero then $\Cov(X_u, X_w) = 0$. Otherwise, since $g_e$'s are constants sign of $\Cov(X_u, X_w)$ will be the same as $\Cov(\mathds{1}_{(u, v) \in \qstar}, \mathds{1}_{(w, v) \in \qstar})$.

Therefore, by applying Proposition~\ref{prop:var_bound} to all $X_u$'s and the fact that they are pairwise negatively correlated we get:
\begin{equation}
\var[Y_v | \Aset] \leq \sum_{u} \var[X_u] \leq \sum_{u} \max(\frac{3x_e}{p_e\cdot t \cdot \delta^2}, \frac{3x_e^2}{p_e \cdot \delta^2}) \leq \sum_{u} \frac{3x_e}{p_e\cdot t \cdot \delta^2} + \sum_{u} \frac{3x_e^2}{p_e \cdot \delta^2}
\end{equation}
For brevity, we are writing $\sum_ {u}$ instead of $\sum_ {(u, v) \in N}$ for all the equations here. To bound the first sum, note that $t = \frac{1}{20 \cdot \epsilon^6 \cdot \delta^2 \cdot p}$ and also $\sum_{(u, v) \in N} x_e \leq 1$ therefore we have:
\begin{equation}
 \sum_{u} \frac{3x_e}{p_e\cdot t \cdot \delta^2} \leq \sum_{u} \frac{3 \cdot 20 \cdot \epsilon^6 \cdot \delta^2 \cdot p_{min} \cdot x_e}{p_e \cdot \delta^2} 
 \leq \sum_{u} 60 \cdot \epsilon^6 \cdot x_e \leq 60 \cdot \epsilon^6
\end{equation}
To bound the second sum, note that for non-crucial edges, we have $x_e \leq \tau$. Since we have $\tau = 20p_{min}\epsilon^5\delta^2$ we get: 

\begin{equation}
 \sum_{u} \frac{3x_e^2}{p_e \cdot \delta^2} \leq \sum_{u} \frac{3 \cdot \tau \cdot x_e}{p_e \cdot \delta^2} 
 \leq \sum_{u} \frac{3 \cdot 20 \cdot \epsilon^{5} \cdot \delta^2 \cdot p_{min}  \cdot x_e}{p_e \cdot \delta^2} 
 \leq \sum_{u} 60 \cdot \epsilon^{5} \cdot x_e \leq 60 \cdot \epsilon^{5}
\end{equation}
Putting things together we get, $\var[Y_v | \Aset] \leq 60 \cdot (\epsilon^{6} + \epsilon^{5})$. Now since we have proved this for any arbitary $\Aset$ we can remove the condition on $\Aset$ and get: 
\begin{equation}
\E[\var[Y_v | S_A]] \leq 60 \cdot (\epsilon^{6} + \epsilon^{5})
\end{equation}
which is exactly \Cref{eq:final_varY} so the proof is complete.
\section{Proof of Lemma~\ref{lem:yvclose1}}\label{sec:yvclosetoone}

Recall the random variable $Y_v$ defined as follows for any vertex $v\in V$.
$$Y_v = \sum_{e=(u,v) \in N} g_{e} \cdot\mathds{1}_{u\in A}  \mathds{1}_{e\in \mathcal{Q}},$$
where $$g_e=\frac{x_e}{p_e\cdot\Pr[e \in Q]\cdot\Pr[\{u, v\} \subset A]}.$$

In this section, we will prove Lemma~\ref{lem:yvclose1}. Recall the statement of the lemma.

\restatelem{\ref{lem:yvclose1}}{\lemyvaboutone{}}

\begin{claim}\label{claim:kjnrrfnj3r}
For any vertex $v\in V$ we have   $$\E[Y_v|v\in A] = \frac{x_v}{\Pr[v\in  A]}$$ where $x_v = \sum_{e \ni v, e \in N} x_e$.
\end{claim}
\begin{proof}
By definition of $Y_v$ we have
\begin{align*}
    \E[Y_v\mid v\in A] &= \E\left[\sum_{e=(u,v) \in N} g_{e} \cdot\mathds{1}_{u\in A}  \mathds{1}_{e\in \mathcal{Q}} \mid v\in A\right] \\& = \sum_{(u,v)\in N}\frac{x_e\Pr[u\in A \mid v\in A]\cdot \Pr[e\in \mc{Q}]}{\Pr[v, u\in A]\cdot \Pr[e\in \mc{Q}]}
\end{align*}
Due to $\Pr[u\in A \mid v\in  A] = \frac{\Pr[u,v\in A]}{\Pr[v\in A]}$ about  conditional expectations we get:  
$$\E[Y_v\mid v\in A]=\sum_{(u,v)\in  N} \frac{x_e}{\Pr[v\in V]}= \frac{x_v}{\Pr[v\in V]}.$$
\end{proof}

\begin{claim}\label{claim:expY}
     For any vertex $v\in  V$ we have:
     $$\E[Y_v]\leq \frac{1}{\delta}$$
\end{claim}
\begin{proof}
    \begin{align*}
       \E[Y_v] &= \sum_{e=(u,v) \in N} \E[g_{e} \cdot\mathds{1}_{u\in A}  \mathds{1}_{e\in \mathcal{Q}}] \\
       &= \sum_{e=(u,v) \in N} \E[\frac{x_e}{p_e\cdot\Pr[e \in Q]\cdot\Pr[\{u, v\} \subset A]} \cdot\mathds{1}_{u\in A}  \mathds{1}_{e\in \mathcal{Q}}] \\
       &= \sum_{e=(u,v) \in N} \E[\frac{x_e}{\Pr[\{u, v\} \subset A]} \cdot\mathds{1}_{u\in A}] \\
       & \leq \sum_{e=(u,v) \in N} \E[\frac{x_e}{\Pr[\{u, v\} \subset A]}] \\
       & \leq \frac{1}{\delta}
    \end{align*}
    The last step is due to Claim~\ref{claim:bothalive} which states $\Pr[\{u, v\} \subset A] \geq \delta$ and $\sum_{e=(u,v) \in N} x_e \leq 1$.
\end{proof}

\begin{claim}\label{claim:YVconditionAdiff}
    Assume we know that 
    $\Pr[|Y_v - \E[Y_v]| \geq \eta] \leq \beta$
    then we have the following bound:
    $$
    \E[Y_v] - 2 \eta \leq \E[Y_v | v \in A] \leq \E[Y_v] + 2\eta
    $$
    
\end{claim}

\newcommand{\NOT}[1]{\ensuremath{\mathrlap{#1}\phantom{#1}}}

\begin{proof}
    Let us define the event $E_b$ to be when $Y_v$ is at least $\eta$ away from its expected value. Formally we have $E_b = \mathds{1}_{ \{ |Y_v - \E[Y_v]| \geq \eta \} }$. From the assumption of the lemma, we know $\Pr[E_b] \leq \beta.$
    
    First, we prove the lower-bound for $\E[Y_v | v \in A]$. We know that:
    \begin{align*}
       \E[Y_v | v \in A] & = \E[Y_v | v \in A \And E_b] \cdot \Pr[E_b | v \in A] + \E[Y_v | v \in A \And \overline{E_b}] \cdot \Pr[\overline{E_b} | v \in A] \\ 
       & \geq  \E[Y_v | v \in A \And \overline{E_b}] \cdot \Pr[\overline{E_b} | v \in A] \\
       & \geq (\E[Y_v] - \eta) \cdot \frac{\Pr[v \in A \cap \overline{E_b}]}{\Pr[v \in A]} \\
       &\geq (\E[Y_v] - \eta) \cdot \frac{\Pr[v \in A] - \Pr[E_b]}{\Pr[v \in A]} \\
        &\geq (\E[Y_v] - \eta) \cdot (1 - \frac{\beta}{\Pr[v \in A]}) \\
        & \geq (\E[Y_v] - \eta) \cdot (1 - \frac{\beta}{\delta}). 
    \end{align*}
      Observation~\ref{claim:expY} states  $\E[Y_v] \geq \frac{1}{\delta}$. Thus, expanding the terms in the last inequality, we get
    \begin{align*}
       \E[Y_v | v \in A] & \geq (\E[Y_v] - \eta) \cdot (1 - \frac{\beta}{\delta}) \\
        & \geq \E[Y_v] - \eta - \frac{\E[Y_v] \cdot \beta}{\delta} \\
        & \geq \E[Y_v] - \eta - \frac{\beta}{\delta^2}.
    \end{align*}
    By picking  $\epsilon\leq \delta^2$ in Table~\ref{table:values}, we get  $\beta \leq \eta \cdot \delta^2$ and 
    $$\E[Y_v | v \in A] \geq \E[Y_v] - 2 \eta.$$
    Now we are ready to prove the desired upper bound. We will again  use the following fact.
    \begin{equation}\label{eq:lem63main}
    \E[Y_v | v \in A]  = \E[Y_v | v \in A \And E_b] \cdot \Pr[E_b | v \in A] + \E[Y_v | v \in A \And \overline{E_b}] \cdot \Pr[\overline{E_b} | v \in A]. 
    \end{equation}
    Note that the second term is upper bounded by $\E[Y_v] + \eta$ as we have $\Pr[\overline{E_b} | v \in A] \leq 1$ and when $E_b$ is not true $Y_v$'s value is always at most $\E[Y_v] + \eta$. To bound the first term we have:
    \begin{align*}
       \E[Y_v | v \in A \And E_b] \cdot \Pr[E_b | v \in A] & \leq
       \E[Y_v | E_b] \cdot \Pr[E_b] 
       \\ & = \E[Y_v] - \E[Y_v | \overline{E_b}] \cdot \Pr[\overline{E_b}]
       \\ & \leq \E[Y_v] - (\E[Y_v] - \eta) \leq \eta
    \end{align*}
    In the last line, we're using the fact that when $\overline{E_b}$ is false, the value of $Y_v$ is at least $\E[Y_v] - \eta$.
    Plugging this in  \eqref{eq:lem63main}, we get
    $$
     \E[Y_v | v \in A] \leq \E[Y_v] + 2 \eta,
    $$
    which completes the proof of this claim.
\end{proof}
Now using the claims we proved in this section, we prove Lemma~\ref{lem:yvclose1}.

\restatelem{\ref{lem:yvclose1}}{\lemyvaboutone{}}

\begin{proof}    
    
    Due to Lemma~\ref{lem:boundedvary} we have
    $$\Pr\Big[|Y_v - \mathbb{E}[Y_v]| \geq \eta \Big] \leq \beta$$
    Furthermore Claim~\ref{claim:kjnrrfnj3r} gives us  $$\E[Y_v|v\in A] = \frac{x_v}{\Pr[v\in  A]}$$
    We begin by proving
    \begin{equation}
     \Pr\left[\left |Y_v-\frac{x_v}{\Pr[v\in  A]}\right |\geq 3\eta\right]\leq \beta
 \end{equation}
 where $x_v = \sum_{e \ni v, e \in N} x_e$. 
 
 This is because for $Y_v$ to have a distance of $3\eta$ from 
 $\frac{x_v}{\Pr[v\in  A]} = \E[Y_v | v \in A]$  has to have a distance of at least $\eta$ from $\E[Y_v]$. This is due to Claim~\ref{claim:YVconditionAdiff} which states that $\E[Y_v | v \in A]$ is between $\E[Y_v] - 2 \eta$ and $\E[Y_v] + 2\eta$.
 Now, since distance of $Y_v$ to $\E[Y_v]$ is at least $\eta$ due to Lemma~\ref{lem:boundedvary}, this happens with probability at most $\beta$.
 
 To complete the proof we need to show $\frac{x_v}{\Pr[v\in  A]}\leq 1$. Recall that we have defined $M_{\mc{O}} = \MM(\mc{H}\cup \mc{N^\star})\cap \mc{H}$, where $\mc{H} = \mc{G}\cap H$ is the actual realization of all the crucial edges, and $\mc{N^\star}$ is a random hallucination of the non-crucial edges containing each edge independently with probability $p_e$. This implies that $\mc{H}\cup \mc{N^\star}$ comes from the same distribution as $\mc{G}$ and as result $M_{\mc{O}}$ is drawn from the same distribution as $\opt$. For any crucial edge $e\in C$ this gives us $\Pr[e\in M_{\mc{O}}]= \Pr[e\in \opt \cap C]$. Summing over crucial edges of a vertex $v$ we get
 \begin{equation}\label{eq:101}
 \Pr[v \in M_{\mc{O}}] = \Pr[v \in \opt \cap C] 
 \end{equation}
 Also based on the definition of $x_v$ we know that 
 \begin{equation}\label{eq:102}
   x_v = \Pr[v \in \opt \cap N]  
 \end{equation}
 Putting things together from \eqref{eq:101} and \eqref{eq:102} we get
  \begin{equation}
   x_v + \Pr[v \in M_{\mc{O}}] = \Pr[v \in \opt \cap N] + \Pr[v \in \opt \cap C] = \Pr[v \in \opt] \leq 1
 \end{equation}
 which means $x_v \leq 1 - \Pr[v \in M_{\mc{O}}] = \Pr[v \notin M_{\mc{O}}] $. Applying property~\ref{viprop:Asideprob} from \cref{def:variance-bounding} which states $\Pr[v\in A] \geq \Pr[v \notin M_{\mc{O}}]$ we get
 $$\frac{x_v}{\Pr[v \in A]} \leq \frac{\Pr[v \notin M_{\mc{O}}]}{\Pr[v\in A]} \leq 1.$$
 Therefore the proof of the lemma is complete. 
\end{proof}
\section{Deferred Proofs}\label{section:proofs}

\begin{claim}
    \label{claim:bothalive} Let $A$ be the subset returned by Algorithm~\ref{alg:Fu-modified}.  For any two vertices $u$ and $v$ that do not share an edge in $H$ we have $\Pr[\{u, v\}\subset A]>\delta$ for a fixed constant $\delta > 0$.
\end{claim}
\begin{proof}
Let $\pi$ be an arbitrary permutation used in Algorithm~\ref{alg:Fu-modified}. We use $\invperm$ for the inverse of the permutation. This means $\invperm (a) = b$ iff $\pi_b = a$. Let us assume $\pi_u < \pi_v$ without loss of generality. Take the three sub-arrays formed on $\invperm$ with $u$ and $v$ on the splits. Formally: $A = [\invperm(1), ..., \invperm(\pi_u - 1)]$, $B = [\invperm(\pi_u + 1), ..., \invperm(\pi_v - 1)]$, $C = [\invperm(\pi_v + 1), ..., \invperm(n)]$.
We claim that with probability, at least $1/36$, vertex $u$, and $v$ will not have an active edge in Algorithm~\ref{alg:Fu-modified} when they arrive. We call this event $E_1$ for future reference. 

Take the sum of $y_e$'s of all of the sequences $A, B, C$ to the vertex $v$ and call them $w_{v,A}, w_{v, B}, w_{v, C}$. Since we know $w_{v,A} + w_{v, B} + w_{v, C} \leq 1$, there exists two of them where their sum is at most $\frac{2}{3}$. Due to symmetry with probability $\frac{1}{3}$ this will be $A$ and $B$. Now define $w_{u, A}, w_{u, B}$ to be sum of $y_e$'s from $u$ to $A, B$. We know that one of $w_{u, A}$ or $w_{u, B}$ is less than $\frac{1}{2}$ since their sum is less than 1. Due to symmetry, with probability $\frac{1}{2}$ we have $w_{u, A} \leq \frac{1}{2}$. Let us call event $E3$ to be $w_A + w_B \leq 2/3$ and event $E4$ to be $w_A \leq 1/2$.
\begin{align*}
       \Pr[E_1] & \geq \Pr[E3] \cdot \Pr[E4]\cdot \Pr[\text{$u$ not active upon arrival} | E4] \cdot \Pr[\text{$v$ not active upon arrival} | E3] \\ 
        & \geq \frac{1}{3} \cdot \frac{1}{2} \cdot \frac{1}{2} \cdot \frac{1}{3} \\
        & \geq \frac{1}{36}.
    \end{align*}

The second to last step is due to Claim~\ref{obs:activelessthanx_e}. This claim implies that the probability that $u$ is not active given $w_{u, A} \leq \frac{1}{2}$ is at least $\frac{1}{2}$. Similarly, the probability that $v$ is not active given $w_{v, A} + w_{v, B} \leq \frac{2}{3}$ is at least $\frac{1}{3}$.

Now consider all vertices $w$ that come after $u$ or $v$ and may have an active edge to them, hence making $u$ or $v$ not being in $A$. For each vertex $w\notin \{u, v\}$, we have  one of these cases. Either $\pi_w < \pi_u, \pi_v$ or $\pi_u < \pi_w < \pi_v$ or $\pi_w > \pi_u, \pi_v$. Due to Claim~\ref{obs:activelessthanx_e} we can see that probability that $u$ or $v$ have an active edge to $w$ is maximized when $\pi_w > \pi_u, \pi_v$ which will make the probability of an active edge to be $g(y_{w, v}) + g(y_{w, u})$. Let $X_w$ be the indicator random variable that shows whether $w$ has an active edge to $u$ or $v$,  hence removing one of them from $A$. We know that $X_w$'s are independent of each other because for each vertex $w$ it decides its active edge only based on its edges to the vertices prior to it in $\pi$ (similar to the argument in Observation~\ref{obs:Fu_properties}). To get a lower-bound for $u$ and $v$ not being in an active edge of all $w$'s with $\pi_w > \pi_v$ (event $E_2$) we get: 

\begin{equation}\label{eq:E2lowerbound}
\Pr[E_2] \geq \prod_{w \neq u, v} \Pr[\text{none of $(w,v)$ or $(w, u)$ become active}] \geq 
\prod_{w \neq u, v} 1 - g(y_{w, u}) - g(y_{w, v})
\end{equation}

The last step is due to Claim~\ref{obs:activelessthanx_e}. To give a lower-bound for $\Pr[E_2]$, let us first see what constraints we have for this optimization problem. For all $w \neq v, u$ we have $y_{w, v} + y_{w, u} \leq 1$. Also, we have $\sum_w y_{v, w} \leq 1$ and $\sum_w y_{u, w} \leq 1$. To simplify our optimization problem, we will combine the two last constraints. That is, the solution to the following optimization problem is a lower-bound for $\Pr[E_2]$:

\begin{align*}
\text{Minimize} \quad &\prod_{w\neq u,v} 1-g(y_{w,u})-g(y_{w,v})&&\\
\text{s.t.} \quad &\sum_{w\neq u,v} y_{w,u}+y_{w,v}\leq 2 &&\\
&y_{w,u}+y_{w,v}\leq 1 &&\forall w\in V\setminus\{u,v\}\\
&y_{w,w'}\geq 0 &&\forall w\in V\setminus\{u,v\} \text{ and } w'\in \{u,v\}
\end{align*}

We claim that the optimal solution for this optimization problem is when there are exactly two vertices $w \neq u, v$ with both  $y_{w, u}=0.5, y_{w, u}=0.5$ and the rest have $y_{w, u}=0, y_{w, u}=0.$ 

We will first prove that for any $w\neq u, v$, in the optimal solution, we have $y_{w, u} = y_{w, v}$ (i.e. this will minimize $1 - g(y_{w, u}) - g(y_{w, v})$). To see this, we claim that if we have fixed $t$ to be $y_{w, u} + y_{w, v}$ the maximum of $g(y_{w, u}) + g(y_{w, v})$ is when  $y_{w, u} = y_{w, v} = \frac{t}{2}$ or $\max_x (g(x) + g(t - x)) = 2g(\frac{t}{2})$.
Suppose that the maximum is in a different point $g(x_1) + g(x_2)$ with $x_1 = \frac{t}{2} - \delta$ and $x_2 = \frac{t}{2} + \delta$ for $\delta > 0$. We can write:

$$g(\frac{t}{2}) = g(x_1) + \int_{x_1}^{\frac{t}{2}} g'(x) dx.$$ 
We also have $$g(x_2) = g(\frac{t}{2}) + \int_{\frac{t}{2}}^{x_2} g'(x) dx.$$ 
Because $g$ has a decreasing derivative on $[0-1]$, 
$\int_{\frac{t}{2}}^{x_2} g'(x) dx < \int_{x_1}^{\frac{t}{2}} g'(x) dx$. Combining all we can see that $2g(\frac{t}{2}) > g(x_1) + g(x_2)$ so we arrive at a contradiction and $\max_x (g(x) + g(t - x)) = 2g(\frac{t}{2})$. As a result, we can safely assume  $y_{w, u} = y_{w, v}$ for all $w\neq u,v$. 
Since we also know $y_{w,u} +y_{w, v}  \leq 1$, this gives us  $$\Pr[X_w = 0] \geq 1 - g(y_{u, w}) - g(y_{v, w}) \geq 1 - 2 \cdot g(0.5) = 0.25$$ 

Now, subject to these constraints, we want to give a lower bound for the probability that for all $w$ we have $X_w = 0$. We can do the following \emph{swap} if we have $w_1, w_2$ such that $0 < y_{u, w_1} < y_{u, w_2} < 0.5$, then by both decreasing $y_{u, w_1}, y_{v, w_1}$ by a small $\Delta>0$ and increasing $y_{u, w_2}, y_{v, w_2}$ by $\Delta$ all the constraints would still be satisfied, and the product of the terms in \Cref{eq:E2lowerbound} would be decreased. 
To see this, let $h(x) = 1 - 2g(x)$. We want to minimize $h(\alpha) h(\beta)$ subject to $\alpha + \beta = s$. To find the minimum point, let us calculate the derivative of $h(x) h(s - x)$.
We know that $h(x) = 1 - 2g(x) = \frac{3 - 4x}{3 + 2x}$. Hence we get 

$$\frac{\partial}{\partial x} h(x) h(s - x) = \frac{\partial}{\partial x} \frac{(3 - 4x) \cdot (3 - 4(s - x))}
{(3 + 2x)\cdot(3 + 2(s - x))} = \frac{36 (4s + 3)(s - 2x)}{(2x + 3)^2 (2s-2x+3)^2}$$
The sign of the terms except for $(s - 2x)$ in the nominator is positive. The sign of the derivative at the point $x = s/2$ changes. For $x < s/2$, it is positive; for $x = s/2$, it is zero; for $x > s/2$, it is negative. This means that the maximum of $h(x)h(s-x)$ is at $x = \frac{s}{2}$, and the minimum is when $x$ is towards $0$ or $s$. This means that the swap described above will indeed decrease the product of terms in equation \eqref{eq:E2lowerbound}.

Therefore, the lower bound for 
$\Pr[E_2]$ is when we have as many pairs of $y_{u, w}, y_{v, w}$ equal to $\frac{1}{2}$. Due to constraint $\sum_{w} y_{u, w} + y_{v, w} \leq 2$, we have at most two vertices $w$ such that $y_{u, w} = y_{v, w} = \frac{1}{2}$. As a result we get
\begin{equation}\label{eq:E2lowerboundfinal}
\Pr[E_2] \geq (1 - 2g(0.5))^2 \geq 0.25^2
\end{equation}

Note that we have $\{u,v\}\subset A$ iff $E_1$ and $E_2$ happen simultaneously. Therefore we get

\begin{equation}\label{eq:E2lowerboundfinal}
\Pr[\{u,v\}\subset A] \geq \Pr[\text{$E_1$ and $E_2$}] \geq \Pr[\text{$E_1$}] \cdot \Pr[\text{$E_2 | E_1$}] \geq \frac{1}{36} \cdot 0.25^2
\end{equation}

This is because our lower bound on $E_2$ is true for any permutation on the arrival of vertices $\pi$ and $\Pr[E2|E1]$ is at least $0.25^2$.
\end{proof}

\restateclaim{\ref{claim:allcrucial}}{\claimallcrucial{}}
\begin{proof}
  An edge $e$  will be in $Q$ if it is in at least one of the matchings $M_1, \dots M_t$. 
  Since these matchings come from  the same distribution as $\opt$ and are independent from each other, we have  
  $$\Pr[e\in Q] = 1 - (1 - x_e)^t \geq 1-e^{-tx_e}\geq 1-e^{-t\tau} \geq 1-1/t\tau\geq 1-\epsilon,$$ which completes the proof. 
\end{proof}

\restateclaim{\ref{claim:lowerbound_einQ}}{\claimlowerboundeinQ{}}
\begin{proof}
An edge $e$  will be in $Q$ if it is in at least one of the matchings $M_1, \dots M_t$. 
  Note that the matchings come from  the same distribution as $\opt$. Therefore, for any  $i, j\in [t]$ we have $\Pr[e\in M_i] = x_e$ and  $\Pr[e\in M_i \text{ and } e\in M_j] = x_e^2$.   Let us first find a lower bound for $\Pr[e\in Q]$ when $tx_e\leq2/3$. Due to the inclusion-exclusion principle, we have 
    $$\Pr[e\in Q] \geq  tx_e - \binom{t}{2} x_e^2 \geq tx_e(1-tx_e) \stackrel{tx_e\leq2/3}{\geq}  tx_e/3 .$$ 
 On the other hand, if  $tx_e>2/3$ we  utilize the following lower-bound which was also used in proving Lemma~\ref{claim:allcrucial}. $$\Pr[e\in Q]= 1-(1-x_e)^t \geq 1-e^{-tx_e} \stackrel{tx_e>2/3}{\geq} 0.48  >  1/3.$$ This completes the proof concluding that $$\Pr[e\in Q]\geq \min(1/3, tx_e/3)$$ holds for any edge $e\in E$.
\end{proof}

\restateclaim{\ref{claim:geupperbound}}{\claimgeupperbound{}}

\begin{proof}
\begin{align*}\label{eq:num14}
  g_e &= \frac{x_e}{\Pr[e \in \mathcal{Q}]\cdot\Pr[\{u, v\} \subset A]}  \\
  &\leq \frac{\tau}{\Pr[e \in \mathcal{Q}] \cdot \Pr[\{u, v\} \subset A] } 
  &&
    \text{(due to $x_e \leq \tau$ since $e$ is non-crucial)}
  \\
  &\leq \frac{\tau}{\Pr[e \in \mathcal{Q}] \cdot \delta } 
  &&
    \text{(due to \cref{claim:bothalive})}
  \\
  &\leq \max(\frac{3 \tau}{\delta}, \frac{1}{\delta t}) 
  &&
  \text{(due to $\Pr[e \in \mathcal{Q}] \geq \min(1/3, tx_e/3)$ (see \cref{claim:lowerbound_einQ}) )
  }
  \\
  &\leq \max(\frac{3 \tau}{\delta}, \frac{\epsilon \tau}{\delta}) \\
  &\leq \frac{\epsilon \tau}{\delta} \\
  &\leq \frac{\epsilon \cdot 20p\epsilon^5\delta^2}{\delta} \\
  &\leq 20\epsilon^6 \delta\\
  &\leq \epsilon^3 
  &&
  \text{(due to choosing an $\epsilon$ sufficiently small s.t. $\epsilon^3\leq 1/(20\delta).$ )}
\end{align*}
This completes the proof of our claim.
\end{proof}

\restatelem{\ref{lemma:integrhjk}}{\lemmaintegrhjk{}}

\begin{proof} 
The main proof ingredient is due to Edmond's celebrated theorem \cite{edmonds1965maximum,schrijver2003combinatorial} and is as follows.

Given a parameter $\xi\in (0,1)$, if for all vertex subsets $U\subset V$ with $3\leq |U|\leq 1/\xi$ the total fractional matching on $U$ is at most $\lfloor\frac{|U|}{2}\rfloor$ then it is possible to round $\bm{f}$ to an integral matching of weight at least  $(1-\xi)\bm{f}\bm{w}.$ These are called the {\em blossom inequalities}. 
That is, for any $U\subset V$ with $|U|\leq 1/\xi$ we should have
$\sum_{\{u, v\}\subset U} f_{(u,v)}\leq \lfloor\frac{|U|}{2}\rfloor.$ For any $|U|\leq 1/\epsilon^{1.5}$ we have  
\begin{align*}
    \sum_{\{u, v\}\subset U} f_{(u,v)}&\leq |U|^2 g_{(u,v)} \leq \left(\frac{1}{\epsilon^{1.5}}\right)^2\epsilon^3\leq 1
    \leq \left\lfloor
    \frac{|U|}{2}\right\rfloor.
\end{align*}
As a result, the above inequalities are satisfied for $\xi=\epsilon^{1.5}$ hence $\bm{f}$ can be rounded to an integral matching with weight at least $(1-\epsilon^{1.5})\bm{f}\bm{w}$. By picking a sufficiently small $\epsilon\leq 1/4$, we get $(1-\epsilon^{1.5})\leq  (1-\epsilon/2)$ completing the proof of this lemma.
\end{proof}

\appendix

\bibliographystyle{alpha}
\bibliography{references}

\end{document}